\documentclass[reqno, eucal,16pt]{amsart}

\usepackage[mathscr]{eucal} 
\usepackage[dvips]{graphicx}
\usepackage[dvips]{color}
\usepackage{amsmath,amsfonts,amssymb,amsthm,amscd}
\usepackage{epic}
\usepackage{eepic}
\usepackage{longtable}
\usepackage{array}
\usepackage{here}
\usepackage{youngtab}

\edef\savecatcodeat{\the\catcode`@}
\catcode`\@=11

\def\@IFNEXTCHAR#1#2#3{\let\@tempe #1\def\@tempa{#2}\def\@tempb{#3}\futurelet
    \@tempc\@IFNCH}
\def\@IFNCH{\ifx \@tempc \@sptoken \let\@tempd\@xifnch
      \else \ifx \@tempc \@tempe\let\@tempd\@tempa\else\let\@tempd\@tempb\fi
      \fi \@tempd}

\def\tb@ifSpecChars#1#2{#1}
\def\tb@ifNoSpecChars#1#2{#2}

\newcount\tb@rpN

\def\tableau{%
  \bgroup
  \@IFNEXTCHAR[{\tb@tableauC}{\tb@tableauC[]}}     

\def\tb@tableauC[#1]{\hbox\bgroup%
    \let\\=\cr
    \def\bl{\global\let\tbcellF\tb@cellNF}%
    \def\tf{\global\let\tbcellF\tb@cellH}
    \dimen2=\ht\strutbox \advance\dimen2 by\dp\strutbox%
    \ifx\baselinestretch\undefined\relax%
    \else%
       \dimen0=100sp \dimen0=\baselinestretch\dimen0%
       \dimen2=100\dimen2 \divide\dimen2 by\dimen0%
    \fi%
    \let\tpos\tb@vcenter
    \tb@initYoung
    \tb@options#1\eoo
    \let\arrow\tb@arrow%
    \dimen0=\Tscale\dimen2%
    \dimen1=\dimen0 \advance\dimen1 by \tb@fframe%
    \lineskip=0pt\baselineskip=0pt
    \def\tb@nothing{}%
    \def\endcellno{$\rss\egroup\bss\egroup}
    \def\endcell{\endcellno\kern-\dimen0}
    \def\begincell{\vbox to\dimen0\bgroup\vss\hbox to\dimen0\bgroup\hss$}%
    \let\overlay\tb@overlay%
    \let\fl\tb@fl%
    \let\fr\tb@fr%
    \let\lss\hss\let\rss\hss\let\tss\vss\let\bss\vss
    \def\mkcell##1{
        \let\tbcellF\tb@cellD
        \def\tb@cellarg{##1}
        \ifx\tb@cellarg\tb@nothing\let\tb@cellarg\tb@cellE\fi%
	        \begincell\tb@cellarg\endcellno
	        \tbcellF
    }%
    \let\savecellF\tbcellF
    \tb@tableauD%
}%

\let\tb@savetableauD\tableauD
{
\gdef\tableauD#1{%
  \tpos{\tabskip=0pt\halign{&\mkcell{##}\cr#1\crcr}}%
  \global\let\tbcellF\savecellF
  \egroup
  \egroup}
}
\let\tb@tableauD\tableauD
\let\tableauD\tb@savetableauD
\let\tb@savetableauD\undefined

\def\tb@options#1{\ifx#1\eoo\relax\else\tb@option#1\expandafter\tb@options\fi}

\def\tb@option#1{%
  \if#1t\let\tpos\tb@vtop\fi
  \if#1c\let\tpos\tb@vcenter\fi
  \if#1b\let\tpos\vbox\fi
  \if#1F\tb@initFerrers\fi
  \if#1Y\tb@initYoung\fi
  \if#1E\tb@initEmpty\fi
  \if#1s\tb@initSmall\fi
  \if#1m\tb@initMedium\fi
  \if#1l\tb@initLarge\fi
  \if#1p\tb@initPartition\fi
  \if#1a\tb@initArrow\fi
}

\def\tb@vcenter#1{\ifmmode\vcenter{#1}\else$\vcenter{#1}$\fi}

\def\tb@vtop#1{\hbox{\raise\ht\strutbox\hbox{\lower\dimen0\vtop{#1}}}}

\def\tb@initPartition{\def\Tscale{.3}}
\def\tb@initSmall{\def\Tscale{1}}
\def\tb@initMedium{\def\Tscale{2}}
\def\tb@initLarge{\def\Tscale{3}}

\def\tb@initArrow{\dimen2=1.25em}

\def\tb@initYoung{%
  \def\tb@cellE{}
  \let\tb@cellD\tb@cellN
  \def\sk{\global\let\tbcellF\tb@cellNF}}
\def\tb@initFerrers{%
  \def\tb@cellE{\bullet}
  \let\tb@cellD\tb@cellNF
  \def\sk{\bullet}}
\def\tb@initEmpty{%
  \def\tb@cellE{}
  \let\tb@cellD\tb@cellNF
  \def\sk{\global\let\tbcellF\tb@cellNF}}

\tb@initMedium

\def\tb@sframe#1{%
  \vbox to0pt{
    \vss
    \hbox to0pt{%
      \hss
      \vbox to\dimen1{
        \hrule depth #1 height0pt
        \vss
        \hbox to\dimen1{
          \vrule width #1 height\dimen1
          \hss
          \vrule width #1
          }%
        \vss
        \hrule height #1 depth 0in
        }%
      \kern-\tb@hframe
      }%
    \kern-\tb@hframe}}

\def\tb@hframe{.2pt}\def\tb@fframe{.4pt}\def\tb@bframe{2pt}
\def\tb@cellH{\tb@sframe{\tb@bframe}}       
\def\tb@cellNF{}                            
\def\tb@cellN{\tb@sframe{\tb@fframe}}       
\let\tbcellF\tb@cellN                       

\def\tb@Fsframe{%
  \vbox to0pt{
    \vss
    \hbox to0pt{%
      \hss
      \vbox to\dimen1{
        \fr@iftop{\hrule depth \fr@width height0pt}{\vskip \fr@width}
        \vss
        \hbox to\dimen1{
	  \fr@ifleft{\vrule width \fr@width height\dimen1}{\hskip \fr@width}
          \hss
          \fr@ifright{\vrule width \fr@width height\dimen1}{\hskip \fr@width}
          }%
        \vss
        \fr@ifbottom{\hrule height \fr@width depth 0in}{\vskip\fr@width}
        }%
      \kern-\tb@hframe
      }%
    \kern-\tb@hframe}}

\def\tb@fr{\@IFNEXTCHAR[{\tb@fra}{\global\let\tbcellF\tb@cellN}}
\def\tb@fra[#1]{%
	\global\let\fr@iftop\tb@IFNO
	\global\let\fr@ifbottom\tb@IFNO%
	\global\let\fr@ifleft\tb@IFNO%
	\global\let\fr@ifright\tb@IFNO%
	\global\let\fr@width\tb@fframe%
	\global\let\tbcellF\tb@Fsframe%
	\froptions#1\eoo
}
\def\froptions#1{\ifx#1\eoo\relax\else\froption#1\expandafter\froptions\fi}
\def\froption#1{
	\if#1t\global\let\fr@iftop\tb@IFYES\fi
	\if#1b\global\let\fr@ifbottom\tb@IFYES\fi
	\if#1l\global\let\fr@ifleft\tb@IFYES\fi
	\if#1r\global\let\fr@ifright\tb@IFYES\fi
	\if#1w\global\let\fr@width\tb@bframe\fi
}
\def\tb@IFYES#1#2{#1}
\def\tb@IFNO#1#2{#2}

\def\tb@rpad{1pt}
\def\tb@lpad{1pt}
\def\tb@tpad{1.8pt}
\def\tb@bpad{1.8pt}

\def\tb@overlay{\endcell\@IFNEXTCHAR[{\tb@overlaya}{\begincell}}
\def\tb@overlaya[#1]{\vbox to\dimen0\bgroup%
  \tb@overlayoptions#1\eoo%
  \tss\hbox to\dimen0\bgroup\lss$}
\def\tb@overlayoptions#1{\ifx#1\eoo\relax\else\tb@overlayoption#1\expandafter\tb@overlayoptions\fi}

\def\tb@overlayoption#1{
  \if#1t\def\tss{\vskip\tb@tpad}\let\bss\vss\fi
  \if#1c\let\tss\vss\let\bss\vss\fi
  \if#1b\def\bss{\vskip\tb@bpad}\let\tss\vss\fi
  \if#1l\def\lss{\hskip\tb@lpad}\let\rss\hss\fi
  \if#1m\let\lss\hss\let\rss\hss\fi
  \if#1r\def\rss{\hskip\tb@rpad}\let\lss\hss\fi
}

\def\tb@fl{\endcell\begincell\vrule depth 0pt width \dimen0 height \dimen0 \endcell\begincell}

\def\tbgobble#1{}
\def\Pscale{1}

\def\skewptn{%
  \@IFNEXTCHAR[{\tb@ptnC}{\tb@ptnC[]}}     

\def\tb@ptnC[#1](#2){%
	{%
    \let\Tscale\Pscale
    \let\\=\cr
   \def\tb@initYoung{%
	\def\tb@cell{\hskip\dimen0\tb@cellN}%
	\def\tb@kernA{\kern.5\dimen0}%
	\def\tb@kernB{\kern-.5\dimen0}%
   }%
   \def\tb@initFerrers{%
	\def\tb@cell{\hbox to\dimen0{\hss$\bullet$\hss}}%
	\def\tb@kernA{}%
	\def\tb@kernB{}%
   }%
    \dimen2=\ht\strutbox \advance\dimen2 by\dp\strutbox%
    \ifx\baselinestretch\undefined\relax%
    \else%
       \dimen0=100sp \dimen0=\baselinestretch\dimen0%
       \dimen2=100\dimen2 \divide\dimen2 by\dimen0%
    \fi%
    \let\tpos\tb@vcenter
    \tb@initYoung
    \tb@options#1\eoo
    \dimen0=\Tscale\dimen2%
    \dimen1=\dimen0 \advance\dimen1 by \tb@fframe%
    \lineskip=0pt\baselineskip=0pt
    \tpos{\skewptnDnewline#2|)}%
	}%
}%

\def\skewptnDnewline#1|{\vbox to\dimen0\bgroup\vss\tb@kernA\hbox\bgroup\skewptnEon#1,|}
\def\skewptnDendline|{\egroup\tb@kernB\vss\egroup\@IFNEXTCHAR{)}{\tbgobble}{\skewptnDnewline}}
\def\skewptnEon#1,{%
	\tb@rpN=#1%
	\ifnum#1>0
	        \loop%
		\tb@cell%
	        \ifnum\tb@rpN>1\advance\tb@rpN by-1%
        	\repeat%
	\fi%
	\@IFNEXTCHAR{|}{\skewptnDendline}{\skewptnEoff}}
\def\skewptnEoff#1,{\hskip #1\dimen0%
	\@IFNEXTCHAR{|}{\skewptnDendline}{\skewptnEon}}

\catcode`\@=\savecatcodeat
\let\savecatcodeat\undefined

\newtheorem{theorem}{Theorem}[section]
\newtheorem{lemma}[theorem]{Lemma}

\newtheorem{proposition}[theorem]{Proposition}

\theoremstyle{definition}

\newtheorem{example}[theorem]{Example}

\newcommand{\Z}{{\mathbb Z}}
\newcommand{\R}{{\mathbb R}}

\newcommand{\ot}{\otimes}
\newcommand{\tab}[1]{\vcenter{\tableau[Fsby]{#1}}}

\begin{document}

\title[Randomized box-ball systems]
{randomized box-ball systems, \\
limit shape of rigged configurations \\
and Thermodynamic Bethe ansatz}

\author{Atsuo Kuniba}
\address{Atsuo Kuniba, Institute of Physics, 
University of Tokyo, Komaba, Tokyo 153-8902, Japan}
\email{\texttt{atsuo.s.kuniba@gmail.com}}

\author{Hanbaek Lyu}
\address{Hanbaek Lyu, Department of Mathematics, 
University of California, Los Angeles, CA 90095, USA}
\email{\texttt{colourgraph@gmail.com}}

\author{Masato Okado}
\address{Masato Okado, Department of Mathematics, Osaka City University, 
Osaka, 558-8585, Japan}
\email{okado@sci.osaka-cu.ac.jp}

\maketitle

\vspace{0.5cm}
\begin{center}{\bf Abstract}
\end{center}
We introduce a probability distribution on the set of states 
in a generalized box-ball system associated with 
Kirillov-Reshetikhin (KR) crystals of type $A^{(1)}_n$.
Their conserved quantities induce $n$-tuple of random Young diagrams 
in the rigged configurations.
We determine their limit shape as the system gets large by
analyzing the Fermionic formula by thermodynamic Bethe ansatz.
The result is expressed as
a logarithmic derivative of a deformed 
character of the KR modules and agrees with the stationary local energy 
of the associated Markov process of carriers.
  
\vspace{0.4cm}

\section{Background and main results}\label{sec:bm}

\subsection{Box-ball systems}
The box-ball system (BBS) \cite{TS} is an integrable cellular automaton
in $1+1$ dimension.
By now it has been generalized widely, and numerous 
aspects have been explored connected to quantum groups, 
crystal base theory (theory of quantum groups at $q=0$), 
solvable lattice models, Bethe ansatz, soliton equations,  
ultradiscretization, tropical geometry and so forth.
See for example a review \cite{IKT} and the references therein.
Here is an example of time evolution $T^{(1)}_\infty$ 
in the 3-color BBS \cite{T93} 
in the notation specified later:

\begin{center}\noindent
\\
$t=0$:\quad 111122221111133211431111111111111111111111111111\vspace{0mm}\\
$t=1$:\quad 111111112222111133214311111111111111111111111111\vspace{0mm}\\
$t=2$:\quad 111111111111222211133243111111111111111111111111\vspace{0mm}\\
$t=3$:\quad 111111111111111122221132433111111111111111111111\vspace{0mm}\\
$t=4$:\quad 111111111111111111112221322433111111111111111111\vspace{0mm}\\
$t=5$:\quad 111111111111111111111112211322433211111111111111\vspace{0mm}\\
$t=6$:\quad 111111111111111111111111122111322143321111111111\vspace{0mm}\\
$t=7$:\quad 111111111111111111111111111221111322114332111111\vspace{0mm}\\
\end{center}

\vspace{0.2cm}\noindent
A letter 1 denotes an empty box whereas $a=2,3,4$ is the one filled with 
a ball with ``color" $a$. 
Initially there are three solitons 
$2222$, $332$ and $43$ with amplitude $4,3,2$.
They proceed to the right
with the velocity equal to the amplitude and eventually 
undergo collisions with messy intermediate states.
However after enough time steps
they come back and separate in exactly the original 
amplitude $2,3,4$ in the reverse order, 
not being smashed into pieces nor glued together\footnote{In $n$-color BBS in general, 
the balls are labeled with $2,3,\ldots, n+1$, and 
a consecutive array of balls $n+1\ge a_1 \ge \cdots \ge a_m \ge 2$
separated sufficiently from other balls 
behaves as a soliton with amplitude = velocity = $m$.
Various choices of $a_1 \ge \cdots \ge a_m$ yield 
internal degrees of freedom of solitons, like quarks in the  
hadrons uud (proton), udd (neutron), uds ($\Lambda$) etc.
The example demonstrates essential features of the soliton scattering;  
interchange of internal degrees of freedom and 
phase shift of asymptotic trajectories.
The final list of solitons is known to be independent of the order of 
collisions of the initial ones (the Yang-Baxter property).}.
This is a manifestation of the integrability, or put more 
practically, existence of 
conserved quantities, either explicit or hidden, governing the dynamics.
The original time evolution $t \rightarrow t+1$ in the $n$-color BBS
was defined by a ball moving algorithm \cite{T93}
as $\mathcal{K}_2\circ\mathcal{K}_3 \circ \cdots \circ \mathcal{K}_{n+1}$,
where $\mathcal{K}_a$ moves every ball with color $a$ once 
starting from the leftmost one successively to its nearest right empty box.
So the number of balls with each color is obviously preserved.
With some effort it is also possible to show that the list of 
amplitude of solitons, if defined properly, 
is also conserved. In the above example it is a partition $(4,3,2)$.
However a quite nontrivial and essential question is; what is 
the {\em complete set} of 
conserved quantities for the general $n$-color BBS?

\subsection{Rigged configuration as action-angle variables}

The answer is known to be an $n$-tuple of Young diagrams.
It was derived from the solution of a more general problem of  
constructing the {\em action-angle variables} of the BBS \cite{KOTY,KOSTY}.
By action variables we mean a set of conserved quantities 
and by the angle variables those {\em linearizing} the dynamics.
The integrability of BBS allows us to transform the system bijectively
into action-angle variables(!)
For the BBS states in the above time evolution, they are combinatorial objects 
that look as follows:

\begin{center}
\begin{picture}(180,48)(50,0)
\setlength{\unitlength}{3.6mm}
\multiput(5,0)(1,0){3}{\line(0,1){3}}
\put(5,0){\line(1,0){2}}
\put(5,1){\line(1,0){3}}
\put(5,2){\line(1,0){4}}
\put(5,3){\line(1,0){4}}
\put(8,1){\line(0,1){2}}
\put(9,2){\line(0,1){1}}
\put(9.3,2.15){$4t$}
\put(8.3,1.15){$6+3t$}
\put(7.3,0.15){$11+2t$}
\put(6.5,3.8){$\mu_1$}
\multiput(14,1)(1,0){2}{\line(0,1){2}}
\put(14,1){\line(1,0){1}}
\multiput(14,2)(0,1){2}{\line(1,0){3}}
\multiput(16,2)(1,0){2}{\line(0,1){1}}
\put(15,3.8){$\mu_2$}
\put(17.3,2.15){1}
\put(15.3,1.15){0}
\multiput(20,2)(1,0){2}{\line(0,1){1}}
\multiput(20,2)(0,1){2}{\line(1,0){1}}
\put(20,3.8){$\mu_3$}
\put(21.3,2.15){0}
\end{picture}
\end{center}

\noindent
For the $n$-color BBS in general, 
there are an $n$-tuple of Young diagrams 
$\mu_1,\ldots, \mu_n$ in which each row is assigned with 
an integer.
The $n$-tuple of Young diagrams and 
the assigned integers are called {\em configuration} and 
{\em rigging}, respectively.
Thus in short, the action-angle variables of the BBS are
{\em rigged configurations} \cite{KOSTY}.
It is the configuration that is preserved and the 
riggings that change linearly in time. 
One indeed sees that the first Young diagram $\mu_1=(4,3,2)$ 
gives the list of amplitude of solitons which remains invariant 
under the time evolution.
The other ones $\mu_2, \ldots, \mu_n$ are ``higher"
conserved quantities reflecting the 
internal degrees of freedom of solitons\footnote{Unfortunately 
to extract them is not so simple and requires a
nested Bethe ansatz (Gelfand-Zetlin) type resolution \cite{KOSTY}.}.
The $n$-color BBS is endowed with the higher time evolutions 
besides the simplest one 
$\mathcal{K}_2\circ\mathcal{K}_3 \circ 
\cdots \circ \mathcal{K}_{n+1}$ mentioned before. 
They are all commutative 
and change the riggings attached to $\mu_2, \ldots, \mu_{n}$ linearly.

Rigged configurations for type $A_n$ has been formulated 
most generally in \cite{KSS}
extending the invention \cite{KKR, KR1} in 1980's.
These works were devoted to a proof of the 
{\em Fermionic formula} for a generalized 
Kostka-Foulkes polynomials (cf. \cite{Ma})
by establishing an elaborate bijection
between rigged configurations 
and other combinatorial objects.
Roughly speaking in the context of BBS, 
the bijection provides the 
direct and inverse scattering maps \cite{KOTY,KOSTY}
\begin{align}
\{\text{BBS states}\} \;\; \longleftrightarrow \;\;
\{\text{action-angle variables}\},
\label{k:sv}
\end{align}
which transform the nonlinear dynamics in BBS to a straight motion.
The $n$-tuple of Young diagrams form a label of the 
{\em iso-level sets} of BBS.
The Fermionic formula tells the {\em multiplicity} of a 
given iso-level set.

\subsection{Randomized BBS and main result}
Now let us embark on a randomized version of the story. 
We assume that some  
probability distribution on the set of BBS states has been introduced.
Then it is natural to ask; 

\begin{enumerate}

\item
What is the probability measure 
on the $n$-tuple of Young diagrams induced by (\ref{k:sv})?

\item
What is the limit shape of them when the system size $L$ of the BBS 
tends to infinity?

\end{enumerate}

\noindent
In this paper we answer (i) and (ii)  
for the most general BBS associated with 
the Kirillov-Reshetikhin (KR) crystals \cite{KMN} of 
the quantum affine algebra $U_q(A^{(1)}_n)$.
The randomness of the BBS states we will be concerned with is
the product of the one at each site. 
The latter is the probability distribution on a single KR 
crystal just proportional to $e^{\mathrm{wt}}$.
(See (\ref{noc}).)
Under this simple choice, the answer to (i) is given by 
the Fermionic form itself multiplied with the Boltzmann factor 
accounting for the $e^{\mathrm{wt}}$ contribution 
as a chemical potential term.
(See (\ref{prom}).)

The Fermionic form measure is quite distinct in nature from 
the well studied ones like 
the Plancherel measure for the symmetric group 
and/or  its Poissonized versions.
It fits an asymptotic analysis by
the {\em thermodynamic Bethe ansatz} (TBA) \cite{YY}.
The method employs the idea of the grand canonical ensemble and 
captures the equilibrium characteristics of the system by 
a variational principle.
The equilibrium condition shows up as the so called TBA equation.
It plays a central role together with the 
equation of state connecting the density of balls with fugacity.
Our TBA analysis is essentially the spectral parameter free 
version of \cite[sec.14]{KNS}. 
In particular the Y-system and the 
Q-system $(Q^{(a)}_i)^2 = Q^{(a)}_{i-1}Q^{(a)}_{i+1}
+ \prod_{b \sim a}Q^{(b)}_i$ for the character 
$Q^{(a)}_i$ of the KR module come into the game naturally.

It turns out that a proper scaling of the Young diagrams 
is to shrink them vertically by $1/L$. 
This feature will be established in \cite{KL} by invoking the large deviation principle. 
The resulting limit shape is described by 
the logarithmic derivative of the 
deformed character of the KR modules as
\begin{align}\label{kk:es5}
\lim_{L \rightarrow \infty}
\frac{1}{L}\bigl(\text{$\#$ of boxes in the left $i$ columns of 
$\mu_a$}\bigr)
= \frac{\partial \log(Q^{(a)}_i \ast Q^{(r)}_s)}{\partial w} 
\Big|_{w=1}.
\end{align}
See (\ref{k:sn}) and (\ref{k:xw}) for the definition of 
the deformed character $Q^{(a)}_i \ast Q^{(r)}_s$
for $(a,i), (r,s) \in [1,n]\times 
\Z_{\ge 1}$\footnote{In general $Q^{(k_1)}_{l_1} \ast \cdots \ast Q^{(k_m)}_{l_m}$ 
is a $w$-deformation of the product $Q^{(k_1)}_{l_1}\cdots Q^{(k_m)}_{l_m}$.  
Representation theoretically,
it is the character of a generalized Demazure module \cite{N}.}.
The data $(r,s)$ is specified according to the choice of 
the set of local states in the BBS\footnote{The original $n$-color BBS \cite{T93}
corresponds to the choice $(r,s)=(1,1)$.}.
The quantity (\ref{kk:es5}) 
coincides with the  
{\em stationary local energy} of a 
carrier in the relevant KR crystal derived in (\ref{k:ep}).
Independent variables in the deformed characters are 
linked to the prescribed fugacity of the BBS by the 
equation of state (\ref{est}) or equivalently (\ref{k:mhn}).
This general and intrinsic answer to the above question (ii)
is our main result in this paper.
Further concrete formulas are available 
for the simplest $n$-color BBS \cite{T93}
in terms of Schur functions in (\ref{k:mst}) and (\ref{k:eai}).
It will be interesting to investigate the results in this paper 
further in the light of recent results on the BBS from 
probabilistic viewpoints \cite{CKST, FNRW, KL, LLP}.

\subsection{Outline of the paper.}
In Section \ref{sec:bbs} we recall basic facts on generalized BBS 
necessary in this paper.
In Section \ref{sec:rbbs} we consider the BBS in a randomized setting.
It amounts to introducing a Markov process of carriers 
associated to each time evolution $T^{(a)}_i$.
We construct a stationary measure of the process quite generally 
by the character of the relevant KR modules (Proposition \ref{pr:st}).
It leads to the stationary local energy  (\ref{k:ep}) or equivalently 
(\ref{k:hhh}).
In Section \ref{sec:ff} we recall the Fermionic formula 
based on \cite{HKOTY,KSS} as a preparation for subsequent sections.
The deformed character in (\ref{k:xm}), (\ref{k:sn})
and its logarithmic derivative will be the building blocks  
in describing the limit shape.
Section \ref{sec:tba} is the main part of the paper.
We identify the Fermionic form with 
the probability measure on the $n$-tuple of 
Young diagrams induced from the randomized BBS via its conserved quantities.
By a TBA analysis, a difference equation characterizing the limit shape of 
the Young diagrams is derived.
Our main result is Theorem \ref{thk:main}, which 
identifies the solution to the difference equation 
with the stationary local energy obtained in Section \ref{sec:rbbs}.
It reveals a new connection between TBA and crystal theory 
via the limit shape problem.
In Section \ref{sec:ex} we deal with the simplest example corresponding to
the $n$-color BBS in \cite{T93}.
The scaled column length of the Young diagrams are 
given explicitly in terms of the Schur function involving the ball densities.
We check the result against the 
stationary local energy of a randomly generated BBS states numerically  
and confirm a good agreement. 
Section \ref{sec:dis} contains a summary and discussion.
We conjecturally describe 
the difference equation and its solution like Theorem \ref{thk:main}
uniformly for the BBS associated with 
the simply-laced quantum affine algebras 
$U_q(\hat{\mathfrak{g}})$
with $\hat{\mathfrak{g}} = A^{(1)}_n, D^{(1)}_n, E^{(1)}_{6,7,8}$.
We also suggest some future problems as concluding remarks.

Throughout the paper we use the notation 
$\theta(\text{true})=1, \theta(\text{false})=0$.

\section{Box-ball systems}\label{sec:bbs}

\subsection{KR crystals}\label{ss:kr}

Consider the classical simple Lie algebra of type 
$A_n$.
We denote its Cartan matrix by $(C_{ab})_{a,b=1}^n$, where
$C_{ab}=C_{ba} = 2\delta_{ab}-\theta(a\sim b)$ and 
$a \sim b$ means that the two nodes $a$ and $b$ are
connected by a bond in the Dynkin diagram, i.e. $|a-b|=1$. 
Let $\varpi_1,\ldots, \varpi_n$ be the fundamental weights and 
$\alpha_1, \ldots, \alpha_n$ be the simple roots.
They are related by $\alpha_a = \sum_{b=1}^n C_{ab}\varpi_b$.
We use the set of positive roots $\Delta_+$,
the weight lattice $P=\bigoplus_{a=1}^n \Z \varpi_a$, 
the root lattice $Q= \bigoplus_{a=1}^n \Z \alpha_a$ 
and their subsets 
$P_+ = \sum_{a=1}^n\Z_{\ge 0} \varpi_a$, 
$Q_+ = \sum_{a=1}^n \Z_{\ge 0} \alpha_a$.
Denote the irreducible module with 
highest weight $\lambda \in P_+$ by $V(\lambda)$ and 
its character by $\mathrm{ch}\,V(\lambda)$.
The latter belongs to 
$\Z[z_1^{\pm 1},\ldots, z_n^{\pm 1}]$ 
where $z_a = e^{\varpi_a}$.

Let $A^{(1)}_n$ be the non-twisted 
affinization of $A_n$ \cite{Kac} and 
$U_q= U_q(A^{(1)}_n)$ be the quantum affine algebra 
(without derivation operator) \cite{D,Ji}.
There is a family of irreducible finite-dimensional representations 
$\{W^{(r)}_s \mid (r,s) \in [1,n]\times \Z_{\ge 0}\}$ of 
$U_q$ called Kirillov-Reshetikhin (KR) module\footnote{The actual  
KR modules carries a spectral parameter. In this paper it is irrelevant and hence 
suppressed.} named after  
the related work on the Yangian \cite{KR2}.
As a representation of $A_n$, $W^{(r)}_s$ is isomorphic to $V(s\varpi_r)$.
$W^{(r)}_s$ is known to have a crystal base $B^{(r)}_s$ \cite{Ka,KMN}.
Roughly speaking, it is a set of basis vectors of a $U_q$-module at $q=0$.
$B^{(r)}_s$ is called a KR crystal. It
is identified with the set of semistandard tableaux
of rectangular shape $(s^r)$ with letters from $\{1,2,\ldots,n+1\}$.
The highest weight element of $B^{(r)}_s$, which is the tableau whose
$j$-th row is all $j$, is denoted by $u^{(r)}_s$.
For two crystals $B_1,B_2$ their tensor product $B_1\otimes B_2$ is
well defined, and as a set 
$B_1\otimes B_2=\{b_1\otimes b_2\mid b_1\in B_1,b_2\in B_2\}$.

Before explaining necessary ingredients related to KR crystals, 
we review a notion of tableau product $S\cdot T$ for two 
tableaux $S,T$. Let $row(T)$ be a row word of a tableau $T$. 
It is obtained by reading letters from bottom to top, left to right in 
each row. Let $row(T)=u_1u_2\cdots u_l$ and we apply to $S$ the row 
bumping algorithm \cite{Fu} successively as
\[
(\cdots((S\leftarrow u_1)\leftarrow u_2)\leftarrow\cdots)\leftarrow u_l.
\]
The resulting tableau is nothing but $S\cdot T$. Alternatively, let
$row(S)=v_1v_2\cdots v_m$ and apply to $T$ the column bumping algorithm
successively as
\[
v_1\rightarrow(v_2\rightarrow(\cdots(v_m\rightarrow T)\cdots)).
\]
The result also gives $S\cdot T$.

We are ready to review the combinatorial $R$ and the (local) energy $H$.
Let $b,c$ be elements of $B^{(a)}_i,B^{(r)}_s$ represented as tableaux.
The combinatorial $R$ is a bijection $R:B^{(a)}_i\otimes B^{(r)}_s
\rightarrow B^{(r)}_s\otimes B^{(a)}_i$ and the energy $H$
is a $\Z$-valued function on $B^{(a)}_i\otimes B^{(r)}_s$ determined 
by the following combinatorial rule \cite{Shi}. 
The image of $R$ is given in such a way
that $R(b\otimes c)=\tilde{c}\otimes\tilde{b}$ is equivalent to
$c\cdot b=\tilde{b}\cdot\tilde{c}$. The fact that for $c\cdot b$ there
is a unique such pair $(\tilde{b},\tilde{c})$ is assured since
the decomposition of the tensor product module 
$V(i\varpi_a)\otimes V(s\varpi_r)$ is multiplicity free.
The value $H(b\otimes c)(=H(\tilde{c}\otimes\tilde{b}))$ is defined
to be the number of nodes strictly 
below the $\max(a,r)$-th row of the tableau
$c\cdot b$. 
By definition, $H$ is nonnegative and $H(u^{(a)}_i\otimes u^{(r)}_s)=0$.
The combinatorial $R$ satisfies the Yang-Baxter equation (cf. \cite{Bax}).

\begin{example}
Consider the $A^{(1)}_3$ case. If $b=\tab{1&2&2&3&4&4}\in B^{(1)}_6,
c=\tab{1&2&2&4}\in B^{(1)}_4$, the tableau product $c\cdot b$, the image of $R$ 
and the value $H(b\otimes c)$ are given by
\[
\tab{1&1&2&2&2&3&4&4\\2&4},\quad
\tab{1&3&4&4}\otimes\tab{1&2&2&2&2&4},\quad
H=2.
\]
If $b=\tab{1&2&3\\2&4&4}\in B^{(2)}_3,c=\tab{1&1\\2&3\\3&4}\in B^{(3)}_2$, 
then $row(b)=214243, \,row(c)=321431$ and they are 
\[
\tab{1&1&1&2&3\\2&2&4&4\\3&3\\4},\quad
\tab{1&1\\2&3\\3&4}\otimes\tab{1&2&2\\3&4&4},\quad
H=1.
\]
\end{example}

\subsection{Deterministic box-ball system}

The original BBS was introduced in \cite{TS}. 
Since then it has been generalized from various viewpoints. 
One of such generalizations was done 
by using KR crystals as we describe below.

Suppose for $b\otimes c\in B^{(a)}_i\otimes B^{(r)}_s$ we have 
$R(b\otimes c)=\tilde{c}\otimes\tilde{b}$. We illustrate it by
\[
\begin{picture}(20,40)(-5,-20)
\put(0,0){\vector(1,0){20}}
\put(10,10){\vector(0,-1){20}}
\put(-7.6,-2){$b$}
\put(7.5,14){$c$}
\put(23,-2){$\tilde{b}$}
\put(7.5,-20){$\tilde{c}$}
\end{picture}
\]
Take a sufficiently large integer $L$ and consider 
$B^{(a)}_i\otimes (B^{(r)}_s)^{\otimes L}$.
Apply the combinatorial $R$ on the $j$-th and $(j+1)$-th component 
(from the left) successively for
$j=1,2,\ldots,L$ to the element $u^{(a)}_i\otimes b_1\otimes b_2\otimes\cdots\otimes b_L$
of $B^{(a)}_i\otimes (B^{(r)}_s)^{\otimes L}$. Let the output be
$b'_1\otimes b'_2\otimes\cdots\otimes b'_L\otimes c$ of 
$(B^{(r)}_s)^{\otimes L}\otimes B^{(a)}_i$. 
Graphically, it can be shown as below.
\begin{align}\label{k:tm}
\begin{picture}(160,40)(-5,-20)
\put(0,0){\vector(1,0){160}}
\multiput(10,10)(20,0){8}{\vector(0,-1){20}}
\put(-17,-2){$u^{(a)}_i$}
\put(7,15){$b_1$}
\put(7,-20){$b'_1$}
\put(27,15){$b_2$}
\put(27,-20){$b'_2$}
\put(147,15){$b_L$}
\put(147,-20){$b'_L$}
\put(164,-2){$c$}
\end{picture}
\end{align}
We call $B^{(a)}_i$ (and its elements) the {\em carrier} 
and $(B^{(r)}_s)^{\otimes L}$ the quantum space. 
The map $T^{(a)}_i:(B^{(r)}_s)^{\otimes L}\rightarrow (B^{(r)}_s)^{\otimes L}$ 
defined by $b_1\otimes b_2 \otimes \cdots \otimes b_L 
\mapsto b'_1\otimes b'_2 \otimes \cdots \otimes b'_L$ is 
called the time evolution operator. Since 
$R(u^{(a)}_i\otimes u^{(r)}_s)=u^{(r)}_s\otimes u^{(a)}_i$, we have 
$T^{(a)}_i((u^{(r)}_s)^{\otimes L})=(u^{(r)}_s)^{\otimes L}$ and $c=u^{(a)}_i$. 
Hence, 
$(u^{(r)}_s)^{\otimes L}$ can be thought of as the vacuum state. If $a=r$ and 
$b_j=u^{(r)}_s$ for $j\ge J$ with such $J$ that $L-J$ is sufficiently large, one
can conclude $c=u^{(a)}_i$. However, if $a\ne r$, $c$ is not always $u^{(a)}_i$,
which means that some particles may be snatched away by $T^{(a)}_i$. To 
prevent such situations, we extend the quantum space $(B^{(r)}_s)^{\otimes L}$
by tensoring $(B^{(a)}_1)^{\otimes L'}$ for sufficiently large $L'$ from right.
Then one can assume $c$ is always $u^{(a)}_i$. We call this extra tensor 
factor the barrier.

Next we recall the conserved quantity under the time evolution 
$T^{(a)}_i$ introduced
for $a=1$ in \cite{FOY} 
for the $n$-color BBS \cite{T93}. 
In order to make $c$ in (\ref{k:tm}) to be $u^{(a)}_i$ we 
attach the barrier if necessary and assume the number of the tensor factors
in the quantum space to be $L$. Define $c_j$ ($j=1,2,\ldots,L$) by
$c_0=u^{(a)}_i,R(c_{j-1}\otimes b_j)=b'_j\otimes c_j$.
The definition corresponds to setting 
the $j$-th vertex in the previous diagram as follows.
\[
\begin{picture}(20,45)(-5,-20)
\put(0,0){\vector(1,0){20}}
\put(10,10){\vector(0,-1){20}}
\put(-21,-2){$c_{j-1}$}
\put(7,14.5){$b_j$}
\put(23,-2){$c_j$}
\put(7,-19.5){$b'_j$}
\end{picture}
\]
We introduce the (row transfer matrix) {\em energy} by 
\begin{equation} \label{ok:Eai}
E^{(a)}_i(b_1\otimes b_2\otimes\cdots\otimes b_L)=\sum_{j=1}^LH(c_{j-1}\otimes b_j).
\end{equation}
One can show $E^{(a)}_i$ is preserved under the time evolution $T^{(a')}_{i'}$,
that is, $E^{(a)}_i(T^{(a')}_{i'}(b))=E^{(a)}_i(b)$ 
following the same argument as \cite{FOY} for $a=a'=1$.
Moreover, supplement of barriers does not change 
$E^{(k)}_l$ for any $k,l$.
We note that these features are valid even when the quantum space 
$(B^{(r)}_s)^{\otimes L}$ is replaced by 
the inhomogeneous one 
$B^{(r_1)}_{s_1} \otimes \cdots \otimes B^{(r_L)}_{s_L}$.

\begin{example}
We give examples of deterministic BBS for $A^{(1)}_3$.
The first one is the case when the carrier is $B^{(1)}_3$ and the
quantum space is $(B^{(1)}_1)^{\ot 13}$.
\[
\begin{picture}(450,40)(-5,-17)
\multiput(20,14)(35,0){13}{\line(0,-1){20}}
\multiput(10,4)(35,0){13}{\line(1,0){20}}
\put(-5,0){111}
\put(30,0){112}
\put(65,0){113}
\put(100,0){111}
\put(135,0){114}
\put(170,0){124}
\put(205,0){224}
\put(240,0){234}
\put(275,0){123}
\put(310,0){223}
\put(345,0){233}
\put(380,0){123}
\put(415,0){112}
\put(450,0){111}
\put(0,2){
\put(17,14){2}
\put(52,14){3}
\put(87,14){1}
\put(122,14){4}
\put(157,14){2}
\put(192,14){2}
\put(227,14){3}
\put(262,14){1}
\put(297,14){2}
\put(332,14){3}
\multiput(367,14)(35,0){3}{1}
}
\put(17,-15){1}
\put(52,-15){2}
\put(87,-15){3}
\put(122,-15){1}
\put(157,-15){1}
\put(192,-15){1}
\put(227,-15){2}
\put(262,-15){4}
\put(297,-15){1}
\put(332,-15){2}
\put(367,-15){3}
\put(402,-15){3}
\put(437,-15){2}
\end{picture}
\]
In general when the carrier is $B^{(1)}_i$ and the quantum space 
is $(B^{(1)}_1)^{\otimes L}$, 
the dynamics on the latter reproduces 
the ball-moving algorithm in the $n$-color BBS \cite{T93} as $i \rightarrow \infty$.

The next example is the case when the carrier is again $B^{(1)}_3$ but the
quantum space is $(B^{(1)}_2)^{\ot 10}$.
\[
\begin{picture}(450,40)(-35,-17)
\multiput(20,14)(35,0){10}{\line(0,-1){20}}
\multiput(10,4)(35,0){10}{\line(1,0){20}}
\put(-5,0){111}
\put(30,0){112}
\put(65,0){124}
\put(100,0){234}
\put(135,0){112}
\put(170,0){222}
\put(205,0){234}
\put(240,0){122}
\put(275,0){113}
\put(310,0){111}
\put(345,0){111}
\put(0,2){
\put(15,14){12}
\put(50,14){24}
\put(85,14){23}
\put(120,14){11}
\put(155,14){22}
\put(190,14){34}
\put(225,14){12}
\put(260,14){13}
\put(295,14){11}
\put(330,14){11}
}
\put(15,-15){11}
\put(50,-15){12}
\put(85,-15){12}
\put(120,-15){34}
\put(155,-15){11}
\put(190,-15){22}
\put(225,-15){34}
\put(260,-15){22}
\put(295,-15){13}
\put(330,-15){11}
\end{picture}
\]
In general choosing the quantum space as 
$(B^{(1)}_s)^{\otimes L}$ corresponds to the boxes 
with capacity $s$.

The last example is the case when the carrier is $B^{(3)}_3$ and the
quantum space is $(B^{(2)}_2)^{\ot 7}$ with the barrier $(B^{(3)}_1)^{\ot 3}$.
\[
\begin{picture}(450,61)(-35,-25)
\multiput(21,14)(36,0){10}{\line(0,-1){20}}
\multiput(11,4)(36,0){10}{\line(1,0){20}}
\put(-4.7,7){\small 111} \put(-4.7,0){\small 222} \put(-4.7,-7){\small 333}
\put(31.8,7){\small 112} \put(31.8,0){\small 223} \put(31.8,-7){\small 334}
\put(67.8,7){\small 122} \put(67.8,0){\small 233} \put(67.8,-7){\small 344}
\put(103.8,7){\small 111} \put(103.8,0){\small 223} \put(103.8,-7){\small 344}
\put(139.8,7){\small 111} \put(139.8,0){\small 222} \put(139.8,-7){\small 344}
\put(175.8,7){\small 112} \put(175.8,0){\small 223} \put(175.8,-7){\small 334}
\put(211.8,7){\small 222} \put(211.8,0){\small 333} \put(211.8,-7){\small 444}
\put(247.8,7){\small 112} \put(247.8,0){\small 223} \put(247.8,-7){\small 444}
\put(283.8,7){\small 111} \put(283.8,0){\small 222} \put(283.8,-7){\small 344}
\put(319.8,7){\small 111} \put(319.8,0){\small 222} \put(319.8,-7){\small 334}
\put(355.8,7){\small 111} \put(355.8,0){\small 222} \put(355.8,-7){\small 333}
\put(16,21){\small 13} \put(16,14){\small 24}
\put(52,21){\small 22} \put(52,14){\small 34}
\put(88,21){\small 11} \put(88,14){\small 23}
\put(124,21){\small 11} \put(124,14){\small 22}
\put(160,21){\small 23} \put(160,14){\small 34}
\put(196,21){\small 22} \put(196,14){\small 44}
\put(232,21){\small 11} \put(232,14){\small 22}
\put(271,28){\small 1}\put(271,21){\small 2}\put(271,14){\small 3}
\put(307,28){\small 1}\put(307,21){\small 2}\put(307,14){\small 3}
\put(343,28){\small 1}\put(343,21){\small 2}\put(343,14){\small 3}
\put(16,-15){\small 11} \put(16,-22){\small 23}
\put(52,-15){\small 12} \put(52,-22){\small 23}
\put(88,-15){\small 22} \put(88,-22){\small 33}
\put(124,-15){\small 11} \put(124,-22){\small 23}
\put(160,-15){\small 12} \put(160,-22){\small 44}
\put(196,-15){\small 11} \put(196,-22){\small 22}
\put(232,-15){\small 22} \put(232,-22){\small 33}
\put(271,-15){\small 2}\put(271,-22){\small 3}\put(271,-29){\small 4}
\put(307,-15){\small 1}\put(307,-22){\small 2}\put(307,-29){\small 4}
\put(343,-15){\small 1}\put(343,-22){\small 2}\put(343,-29){\small 4}
\end{picture}
\]
This is the most general situation.
Local states and carriers are no longer simple boxes but 
possess a structure of a {\em shelf} with a nontrivial constraint 
on the arrangement of balls from the 
semistandard condition of the tableaux.

\end{example}

Introduction of carriers \cite{TM} as a hidden dynamical variable of BBS 
was a corner stone in the development of the theory.
It provided the apparently nonlocal ball moving algorithm 
with a {\em local} description encoded in a {\em single} vertex 
in the above diagrams.
A further discovery that these vertices are nothing but the 
combinatorial $R$ unveiled the nature of BBS as solvable vertex models \cite{Bax}
at $q=0$, where time evolutions are naturally identified with 
commuting row transfer matrices \cite{HHIKTT,FOY}.
As we will see in Section \ref{sec:rbbs},
carriers also play a fundamental role in the randomized BBS
via their {\em Markov processes}.

\subsection{Rigged configuration as action angle variables}\label{ss:eca}

Here we review a combinatorial object called rigged configuration and see how it is
used to linearize the BBS dynamics. 
Rigged configurations are defined based on data
$\{(k_j,l_j)\}_{1\le j\le L}$ such that $(k_j,l_j)\in [1,n]\times\Z_{\ge1}$. 
Through the 
Kirillov-Schilling-Shimozono (KSS) bijection which we discuss later, it is related to
the tensor product of KR crystals 
$B^{(k_1)}_{l_1}\otimes\cdots\otimes B^{(k_L)}_{l_L}$.
A rigged configuration consists of a configuration, an $n$-tuple of Young diagrams
$\mu_1,\ldots,\mu_n$, and riggings, sequence of nonnegative integers 
attached to each row of $\mu_a$ for $a\in [1,n]$. 
Let $m^{(a)}_i$ be the number of rows of length $i$ in $\mu_a$.
Define 
$e^{(a)}_i$ and {\em vacancy} $p^{(a)}_i$ by 
\begin{align}
p^{(a)}_i  &=  \sum_{j=1}^L \delta_{a,k_j}\min(i,l_j)
-\sum_{b=1}^nC_{ab}e^{(b)}_i,
\label{ok:pai}\\
e^{(a)}_i &= \sum_{j \ge 1}\min(i,j)m^{(a)}_j.
\label{ok:eai}
\end{align}
A configuration is required to satisfy $p^{(a)}_i\ge0$ for any 
$(a,i)\in [1,n]\times\Z_{\ge1}$ and riggings of the rows of length $i$ in $\mu_a$
not to exceed $p^{(a)}_i$. Among riggings of the rows of the same length in $\mu_a$,
the order does not matter. So we label riggings in non increasing order when going
downwards. 
From these definitions one can  immediately write down the 
number of the rigged configurations with the 
prescribed configuration $\mu_1,\ldots,\mu_n$ as
\begin{align}\label{k:be}
\prod_{1 \le a \le n, i\ge 1}
\left({p^{(a)}_i +  m^{(a)}_i \atop m^{(a)}_i}\right).
\end{align}
This is an ultimate generalization of the celebrated Bethe formula 
\cite[eq.(45)]{Be} due to \cite{KKR, KR1,KSS} for type $A^{(1)}_n$.
See \cite[sec.1]{HKOTT} for a historical account and  
\cite[sec.13]{KNS} for a concise review.
We will come back to this {\em Fermionic form} as the main object of the 
TBA analysis in Section \ref{sec:tba}.

The KSS bijection \cite{KSS} gives an algorithm to construct
an element of the tensor product of KR crystals 
$B=B^{(k_1)}_{l_1}\otimes\cdots\otimes B^{(k_L)}_{l_L}$
from a rigged configuration. The image of this bijection consists of special elements
which we call {\em highest} states. Representation theoretically, they correspond to
highest weight vectors of $B$. It is equivalent to saying that the tableau product 
$b_L\cdot\cdots\cdot b_1$ ($b_j\in B^{(k_j)}_{l_j})$ is a tableau such that letters 
in the $i$-th row are all $i$. 

The KSS bijection 
separates the BBS states into action and angle variables. 
It is known \cite{KOSTY} that if $b$ is a highest state, then the application
of $T^{(a)}_i$ causes, in the rigged configuration side, the increase of riggings
by $\delta_{ac}\min(i,l)$ when they are attached to the length $l$ row of $\mu_c$.

Identifying rigged configurations originating in the Bethe ansatz \cite{Be}
with action-angle variables of BBS 
implies a correspondence between Bethe strings in the former and 
solitons in the latter. 
This is natural as we will also comment in the end of 
Section \ref{ss:ff}.
As far as the action variables are concerned,
this {\em soliton/string correspondence} \cite{KOTY, KOSTY} 
is quantified most generally 
as \cite{Sa}
\begin{align}\label{k:ii}
E^{(a)}_i = e^{(a)}_i.
\end{align}
Remember that the LHS is the row transfer matrix energy 
in \eqref{ok:Eai}, which was indeed known (for $a=1$) to measure 
the amplitude of solitons \cite{FOY} for the original 
$n$-color BBS \cite{T93}.
The RHS is defined by \eqref{ok:eai} from the rigged configuration 
which is essentially an assembly of Bethe strings \cite{Be,KKR,KR1}.
Thus the LHS and the RHS in (\ref{k:ii})
are referring to solitons and strings, respectively.
Our main result Theorem \ref{thk:main} in this paper 
may be regarded as a generalization of (\ref{k:ii}) 
to a randomized situation.

\begin{example}
We give examples of the KSS bijection for $A^{(1)}_3$. An element of
$(B^{(1)}_1)^{\ot20}$
\[
1\ot2\ot1\ot2\ot3\ot4\ot1\ot1\ot3\ot2\ot1\ot1\ot2\ot3\ot2\ot1\ot
3\ot4\ot4\ot1
\]
is a highest state which corresponds to the rigged configuration below.
The numbers left to the Young diagram are vacancies. 
$p^{(1)}_2=3, p^{(1)}_1=6, p^{(2)}_2=1, p^{(2)}_1=3, 
p^{(3)}_2=1, p^{(3)}_1=0$.
\[
\unitlength 10pt
\begin{picture}(30,9)(-8,-2)
\Yboxdim{10pt}
\put(1,-2){\yng(2,2,2,1,1,1,1,1,1)}
\put(0.2,6){3}
\put(0.2,3){6}
\put(3.2,6){2}
\put(3.2,5){1}
\put(3.2,4){1}
\put(2.2,3){4}
\put(2.2,2){3}
\multiput(2.2,1)(0,-1){4}{0}
\put(7,3){\yng(2,2,2,1)}
\put(6.2,6){1}
\put(6.2,3){3}
\put(9.2,6){1}
\put(9.2,5){1}
\put(9.2,4){0}
\put(8.2,3){1}
\put(13,5){\yng(2,1)}
\put(12.2,6){1}
\put(12.2,5){0}
\put(15.2,6){1}
\put(14.2,5){0}
\end{picture}
\]
Similarly, 
\[
\tab{1&1\\2&2}\ot\tab{1&1\\2&3}\ot\tab{1&2\\3&4}\ot\tab{1&1\\3&3}
\ot\tab{1&1\\2&2}\ot\tab{1&1\\2&4}\ot\tab{1&1\\2&2}\ot\tab{2&3\\4&4}
\ot\tab{1&2\\3&3}\ot\tab{1&3\\2&4}
\]
is a highest state of $(B^{(2)}_2)^{\ot10}$ which corresponds to the following one.
\[
\unitlength 10pt
\begin{picture}(30,6)(-5,0)
\Yboxdim{10pt}
\put(1,4){\yng(4,1)}
\put(0.2,5){3}
\put(0.2,4){2}
\put(5.2,5){2}
\put(2.2,4){0}
\put(9,0){\yng(5,4,1,1,1,1)}
\put(8.2,5){4}
\put(8.2,4){6}
\put(8.2,3){2}
\put(14.2,5){3}
\put(13.2,4){0}
\multiput(10.2,3)(0,-1){3}{2}
\put(10.2,0){0}
\put(18,4){\yng(3,2)}
\put(17.2,5){0}
\put(17.2,4){0}
\put(21.2,5){0}
\put(20.2,4){0}
\end{picture}
\]
\end{example}

\section{Randomized box-ball system}\label{sec:rbbs}

\subsection{Markov process of carrier}
Now we introduce a randomized version of BBS.
Let 
$\pi^{(r)}_s: B^{(r)}_s \rightarrow \R_{>0}$ be a probability distribution
meaning that $\sum_{b \in B^{(r)}_s}\pi^{(r)}_s(b) = 1$.
We consider the ensemble of the states of the BBS on 
$(B^{(r)}_s)^{\otimes L}$ in which the local states 
are independent and identically distributed (i.i.d.)
according to $\pi^{(r)}_s$.
Taking them as the initial condition, we apply 
a time evolution $T^{(a)}_i$.
The carrier from $B^{(a)}_i$ proceeds to the right 
interacting with the random local states 
successively by the combinatorial $R$ on $B^{(a)}_i \otimes B^{(r)}_s$.
The bombardment by the random local states induces a stochastic process 
of the carrier. 
It is the Markov process on $B^{(a)}_i$ whose transition rate is 
specified as
\begin{align}\label{k:mp}
\mathrm{Rate}(u \rightarrow u') = 
\sum_{b,b' \in B^{(r)}_s, R(u\otimes b) = b' \otimes u'}
\pi^{(r)}_s(b)\qquad
(u, u' \in B^{(a)}_i).
\end{align} 
The condition on the sum is depicted as a vertex in (\ref{k:tm}) as
\[
\begin{picture}(20,40)(-5,-20)
\put(0,0){\vector(1,0){20}}
\put(10,10){\vector(0,-1){20}}
\put(-8.5,-2){$u$}
\put(7.5,13.5){$b$}
\put(23,-2){$u'$}
\put(7.5,-20){$b'$}
\end{picture}
\]
We have 
$\sum_{u' \in B^{(a)}_i}\mathrm{Rate}(u \rightarrow u')= 
\sum_{b \in B^{(r)}_s}\pi^{(r)}_s(b) = 1$ indeed.
In \cite{KL} it has been shown that this Markov process is irreducible and 
has the unique stationary measure for $r=s=1$. 
We conjecture and assume the irreducibility for general $r,s$ in what follows. 
Denote the resulting stationary measure by 
$\tilde{\pi}^{(a)}_i: B^{(a)}_i \rightarrow \R_{>0}$.

\begin{example}\label{ex:ca}
Consider $A^{(1)}_2$ BBS with local states from  
$B^{(r)}_s=B^{(1)}_1=\{1,2,3\}$ 
and carrier from $B^{(a)}_i=B^{(1)}_2=\{11, 12, 13, 22, 23, 33\}$.
Take  
$\pi^{(r)}_s: B^{(r)}_s \rightarrow \R_{>0}$
as $\pi^{(1)}_1(b)= p_b \,(b=1,2,3)$ with some $p_b$ obeying
$p_1+p_2+p_3=1$.
The combinatorial $R$ is given by
\begin{align*}
11 \otimes 1 & \mapsto 1 \otimes 11
\qquad
11 \otimes 2 \mapsto 1 \otimes 12
\qquad
11 \otimes 3 \mapsto 1 \otimes 13,\\
12 \otimes 1 & \mapsto 2 \otimes 11
\qquad
12 \otimes 2 \mapsto 1 \otimes 22
\qquad
12 \otimes 3 \mapsto 2 \otimes 13,\\
13 \otimes 1 & \mapsto 3 \otimes 11
\qquad
13 \otimes 2 \mapsto 1 \otimes 23
\qquad
13 \otimes 3 \mapsto 1 \otimes 33,\\
22 \otimes 1 & \mapsto 2 \otimes 12
\qquad
22 \otimes 2 \mapsto 2 \otimes 22
\qquad
22 \otimes 3 \mapsto 2 \otimes 23,\\
23 \otimes 1 & \mapsto 3 \otimes 12
\qquad
23 \otimes 2 \mapsto 3 \otimes 22
\qquad
23 \otimes 3 \mapsto 2 \otimes 33,\\
33 \otimes 1 & \mapsto 3 \otimes 13
\qquad
33 \otimes 2 \mapsto 3 \otimes 23
\qquad
33 \otimes 3 \mapsto 3 \otimes 33.
\end{align*}
By the definition the transitions in the first, second and the third column
happen with probabilities $p_1, p_2, p_3$, respectively.
Thus denoting $\tilde{\pi}^{(1)}_2(23)=\tilde{\pi}_{23}$ etc for short, the 
stationary condition reads
\begin{align*}
\tilde{\pi}_{11} &= 
p_1 \tilde{\pi}_{11} +
p_1 \tilde{\pi}_{12} +
p_1 \tilde{\pi}_{13},\\
\tilde{\pi}_{12} &= 
p_1 \tilde{\pi}_{22} +
p_1 \tilde{\pi}_{23} +
p_2 \tilde{\pi}_{11},\\
\tilde{\pi}_{13} &= 
p_1 \tilde{\pi}_{33} +
p_3 \tilde{\pi}_{11} +
p_3 \tilde{\pi}_{12},\\
\tilde{\pi}_{22} &= 
p_2 \tilde{\pi}_{12} +
p_2 \tilde{\pi}_{22} +
p_2 \tilde{\pi}_{23},\\
\tilde{\pi}_{23} &= 
p_2 \tilde{\pi}_{13} +
p_2 \tilde{\pi}_{33} +
p_3 \tilde{\pi}_{22},\\
\tilde{\pi}_{33} &= 
p_3 \tilde{\pi}_{13} +
p_3 \tilde{\pi}_{23} +
p_3 \tilde{\pi}_{33}.
\end{align*} 
Under the assumption $p_1+p_2+p_3=1$, these equations admit
a unique solution $(\tilde{\pi}_{ij})$ such that
$\sum_{1\le i \le j \le 3}\tilde{\pi}_{ij}=1$.
For instance let us parametrize $p_i$ as
\begin{align}\label{k:ps}
p_1 = \frac{z_1^2z_2}{z_1+z_1^2z_2+z_2^2},\quad
p_2 = \frac{z_2^2}{z_1+z_1^2z_2+z_2^2},\quad
p_3 = \frac{z_1}{z_1+z_1^2z_2+z_2^2}.
\end{align}
Then $\tilde{\pi}_{ij}$ is given by
\begin{align}
\tilde{\pi}_{11} &= \frac{z_1^4z_2^2}{Q},\quad
\tilde{\pi}_{12} = \frac{z_1^2z_2^3}{Q},\quad
\tilde{\pi}_{13} = \frac{z_1^3z_2}{Q},
\quad
\tilde{\pi}_{22} = \frac{z_2^4}{Q},\quad
\tilde{\pi}_{23} = \frac{z_1z_2^2}{Q},\quad
\tilde{\pi}_{33} = \frac{z_1^2}{Q}
\end{align}
with $Q = z_1^2+z_1^3 z_2 + z_1z_2^2+ z_1^4 z_2^2+z_1^2 z_2^3+z_2^4$.

\end{example}

\subsection{Stationary measure of carrier}

To generalize Example \ref{ex:ca} is straightforward.
The stationary measure $\tilde{\pi}^{(a)}_i$ is determined by the following 
stationary condition of the carrier process:
\begin{align}\label{stc}
\tilde{\pi}^{(a)}_i(u')
= \sum_{u \in B^{(a)}_i}\tilde{\pi}^{(a)}_i(u)
\mathrm{Rate}(u \rightarrow u') = 
\sum_{u \otimes b \in B^{(a)}_i \otimes B^{(r)}_s, 
R(u \otimes b) \in B^{(r)}_s \otimes u'}
\tilde{\pi}^{(a)}_i(u)\pi^{(r)}_s(b),
\end{align}
where the latter equality follows from (\ref{k:mp}).
The carrier and local states are taken from 
$B^{(a)}_i$ and $B^{(r)}_s$, respectively.

For a partition 
$\lambda=(\lambda_1,\lambda_2,\ldots, \lambda_{n+1})$, 
let $s_\lambda(w_1,\ldots, w_{n+1})$ denote the 
associated Schur polynomial \cite{Ma}:
\begin{align}
s_\lambda(w_1,\ldots, w_{n+1})
= \frac{\det(w_k^{\lambda_j+n+1-j})_{j,k=1}^{n+1}}
{\det(w_k^{n+1-j})_{j,k=1}^{n+1}}.
\label{k:sc}
\end{align}
This is the well-known Weyl formula
for the character $\mathrm{ch}V(\lambda)$
under the identification of $\lambda$ with 
$\sum_{i=1}^n(\lambda_i-\lambda_{i+1})\varpi_i \in P_+$.
We use a special notation when $\lambda$ is a rectangle.
\begin{align}
&Q^{(a)}_i  = \mathrm{ch}\, V(i \varpi_a) 
= \sum_{b \in B^{(a)}_i}e^{\mathrm{wt}(b)}
= s_{(i^a)}(w_1,\ldots, w_{n+1}),
\label{qf}\\
&w_j = z^{-1}_{j-1}z_j\;\; (1 \le j \le n+1),\qquad z_0=z_{n+1}=1.
\label{k:wde}
\end{align}

The following proposition gives an explicit expression for the 
stationary measure $\tilde{\pi}^{(a)}_i$ under a particular choice of the
measure $\pi^{(r)}_s$ for the local states. 
\begin{proposition}\label{pr:st}
The choice 
\begin{align}\label{noc}
\tilde{\pi}^{(a)}_i(u) 
&= \pi^{(a)}_i(u) = \frac{\exp\mathrm{wt}(u)}{Q^{(a)}_i}
\in \R(z_1,\ldots, z_n)\qquad (u \in B^{(a)}_i)
\end{align}
for any $(a,i) \in [1,n]\times \Z_{\ge 1}$
satisfies the stationary condition (\ref{stc})
and the normalization condition
$\sum_{u \in B^{(a)}_i}\pi^{(a)}_i(u) = 1$.
\end{proposition}
\begin{proof}
Consider the combinatorial $R$
\begin{align}
R:\;\; B^{(a)}_i \otimes B^{(r)}_s &\rightarrow B^{(r)}_s \otimes B^{(a)}_i\\
 u \otimes b &\mapsto b' \otimes u'.
\end{align}
Take the $\exp \mathrm{wt}$ of the both sides and sum over $b' \in B^{(r)}_s$ fixing 
$u' \in B^{(a)}_i$.
The result reads
\begin{align}
Q^{(a)}_iQ^{(r)}_s\sum_{u \otimes b \in B^{(a)}_i \otimes B^{(r)}_s, 
R(u \otimes b) \in B^{(r)}_s \otimes u'}
\pi^{(a)}_i(u)
\pi^{(r)}_s(b) = Q^{(r)}_s \exp\mathrm{wt}(u').
\end{align}
This yields (\ref{stc}) with $\tilde{\pi}^{(a)}_i = \pi^{(a)}_i$
by dividing by $Q^{(a)}_iQ^{(r)}_s$.
The normalization condition is obvious from (\ref{qf}).
\end{proof}

Proposition \ref{pr:st} tells that as long as the randomness 
$\pi^{(r)}_s$ of the local states are taken to be proportional to 
$e^{\mathrm{wt}}$ as in (\ref{noc}), 
the stationary measure $\tilde{\pi}^{(a)}_i$ of the carrier for $T^{(a)}_i$ 
is independent of the choice of the KR crystal $B^{(r)}_s$, and 
it is also given by the same formula.
This is a reminiscent of the integrability 
of the original (deterministic) BBS.  
For $A^{(1)}_2$, the KR module $W^{(1)}_s$ is isomorphic to the 
degree $s$ symmetric tensor representation of $sl_3$, and  
(\ref{noc}) indeed reproduces Example \ref{ex:ca} if $z_a$ is identified with 
$e^{\varpi_a}$.

\subsection{Stationary local energy}\label{ss:sle}

In view of Proposition \ref{pr:st} we write 
the probability distribution $\tilde{\pi}^{(a)}_i$ of the carrier 
also as $\pi^{(a)}_i$ from now on.
It contains $n$ real positive parameters.

Now let us calculate the row transfer matrix energy per site 
in the stationary state:
\begin{align}\label{k:hai}
h^{(a)}_i = \lim_{L \rightarrow \infty}\frac{1}{L}E^{(a)}_i.
\end{align} 
See (\ref{ok:Eai}) for the definition of $E^{(a)}_i$.
On average 
the carrier $x \in B^{(a)}_i$ and a local state $y \in B^{(r)}_s$
arrive at a vertex with the probability 
$\pi^{(a)}_i(x)\pi^{(r)}_s(y)$.
Their encounter produces the local energy $H(x \otimes y)$.
Thus we have
\begin{align}\label{k:ep}
h^{(a)}_i = 
\sum_{x \otimes y \in B^{(a)}_i\otimes B^{(r)}_s}
H(x\otimes y)\pi^{(a)}_i(x)\pi^{(r)}_s(y).
\end{align}
We call this {\em stationary local energy} for the carrier 
from $B^{(a)}_i$ or for the time evolution $T^{(a)}_i$.
Note that $h^{(a)}_i$ does depend on the choice of the set of local states 
$B^{(r)}_s$ although it is suppressed in the notation.  

To write (\ref{k:ep}) more concretely,
consider the irreducible decomposition of the tensor product
$V(i\varpi_a)\otimes V(s\varpi_r)$.
It is multiplicity free as noted in Section \ref{ss:kr}, and 
results in the identity of the Schur functions as 
\begin{align*}
s_{(i^a)}(w_1,\ldots,w_{n+1})
s_{(s^r)}(w_1,\ldots,w_{n+1})
= \sum_{\nu \in \mathcal{P}^{(a,r)}_{i,s}}s_{\nu}(w_1,\ldots,w_{n+1}).
\end{align*}
Here $\mathcal{P}^{(a,r)}_{i,s}$ denotes the set of partitions
(Young diagrams) labeling the irreducible components described by 
the Littlewood-Richardson rule.
Concretely one has
\[
\mathcal{P}^{(a,r)}_{i,s} 
= \{\nu=(\nu_i): \text{partition}\mid
\ell(\nu)\le \min(n+1,a+r), \, \nu \supseteq (i^a),\,  
\nu \supseteq (s^r), \, |\nu| = ia+sr\},
\]
where $\ell(\nu)$ denotes the length of the partition $\nu$.
From the description of the local energy $H$ in Section \ref{ss:kr},
the result (\ref{k:ep}) is expressed as
\begin{align}\label{k:hhh}
h^{(a)}_i = 
\frac{\sum_{\nu \in \mathcal{P}^{(a,r)}_{i,s}}
(\sum_{j > \max(a,r)}\nu_j)s_{\nu}(w_1,\ldots,w_{n+1})}
{s_{(i^a)}(w_1,\ldots, w_{n+1})s_{(s^r)}(w_1,\ldots, w_{n+1})}.
\end{align}

\begin{example}\label{ex:h}
Consider the case $r=1$ for the set $B^{(r)}_s$ of local states.
Then
$\mathcal{P}^{(a,1)}_{i,s} = \{
(i+s-k,i^{a-1},k)\mid k \in [0,\min(i,s)]\}$
for $x \otimes y \in B^{(a)}_i\otimes B^{(1)}_s$.
The local energy takes the value $H(x\otimes y) =k$.
Thus $h^{(a)}_i$ (\ref{k:hhh}) reads as 
\begin{align}
h^{(a)}_i = \frac{\sum_{k=1}^{\min(i,s)}
k\, s_{(i+s-k,i^{a-1},k)}(w_1,\ldots, w_{n+1})}
{s_{(i^a)}(w_1,\ldots, w_{n+1})s_{(s)}(w_1,\ldots, w_{n+1})}.
\end{align}
We will use this formula with $s=1$ in Section \ref{sec:ex}.
\end{example}

\section{Fermionic form}\label{sec:ff}

\subsection{Deformed character}
Given a tensor product
$B = B^{(k_1)}_{l_1} \otimes \cdots \otimes B^{(k_L)}_{l_L}$,
we set
\begin{align}
\chi_w(B) &= \sum_{b_1 \otimes \cdots \otimes b_L \in B}
w^{D(b)}e^{\mathrm{wt}(b)}
\in \Z_{\ge 0}[z^{\pm 1}_1,\ldots, z^{\pm 1}_n, w],
\label{k:xw}\\
D(b_1 \otimes \cdots \otimes b_L)  &= 
\sum_{1 \le i < j \le L}H(b_i \otimes b^{(i+1)}_j) \in \Z_{\ge 0}.
\label{k:d}
\end{align}
Here $w$ is a parameter having nothing to do with 
$w_1, \ldots, w_{n+1}$ in (\ref{k:wde}).
The element $b^{(r)}_j \in B^{(k_j)}_{l_j}\, (j \ge r)$ is the one 
occurring at the position $r$ by  
sending $b_j \in B^{(k_j)}_{l_j}$ to the left by successively applying 
the combinatorial $R$ as
\begin{equation}\label{brj}
\begin{split}
b_r \otimes b_{r+1} \otimes \cdots 
 \otimes b_{j-1} \otimes b_j
&\simeq 
b_r  \otimes b_{r+1} \otimes \cdots 
\otimes b^{(j-1)}_{j} \otimes b'_{j-1}\simeq \cdots\\
&\simeq 
b_r \otimes b^{(r+1)}_j \otimes\cdots 
\otimes b'_{j-2} \otimes b'_{j-1}\\
&\simeq 
b^{(r)}_j  \otimes b'_{r} \otimes\cdots 
\otimes b'_{j-2} \otimes b'_{j-1}.
\end{split}
\end{equation} 
In particular we set $b^{(r)}_r=b_r$.
The procedure is depicted as

\begin{picture}(200,75)(-180,-33)

\put(-18,22){$b_r$} 
\put(5,22){$b_{r+1}$}\put(31,22){$\cdots$}
\put(53,22){$b_{j-2}$}
\put(80,22){$b_{j-1}$}\put(105,-3){$b_j$}

\put(-18,-27){$b'_r$} 
\put(5,-27){$b'_{r+1}$}\put(31,-27){$\cdots$}
\put(53,-27){$b'_{j-2}$}
\put(81,-27){$b'_{j-1}$}\put(-54,-3){$b^{(r)}_j$}

\put(100,0){\vector(-1,0){137}}
\multiput(-14,-1)(25,0){2}{
\put(0,15){\vector(0,-1){30}}
}
\multiput(61,-1)(25,0){2}{
\put(0,15){\vector(0,-1){30}}
}
\end{picture}

In contrast to the energy associated with the row transfer matrices 
(\ref{ok:Eai}), 
the quantity $D$ in (\ref{k:d}) corresponds to the energy 
of a {\em corner transfer matrix}\footnote{For types other than $A_n$,
one needs to add boundary energy. See (\ref{k:db}).}, 
which goes back to \cite[chap.13]{Bax}.
In fact, using the Yang-Baxter equation for the combinatorial $R$
it can be identified with 
the sum of the local energy associated to all the $L(L-1)/2$ 
vertices in the following diagram ($L=3$ example).

\begin{picture}(130,70)(-130,37)

\put(0,-9){
\put(100,100){$b_1$}
\put(100,80){$b_2$}
\put(100,60){$b_3$}
}

\put(95,95){\line(-1,0){10}}\put(85,95){\vector(0,-1){55}}
\put(95,75){\line(-1,0){30}}\put(65,75){\vector(0,-1){35}}
\put(95,55){\line(-1,0){50}}\put(45,55){\vector(0,-1){15}}

\end{picture}

This quadrant structure is essentially a combinatorial 
counterpart of \cite[Fig.13.1(b)]{Bax}.
By the definition we have
\begin{align}\label{k:kq}
\chi_w(B)|_{w=1}= \prod_{i=1}^L
\sum_{c \in B^{(k_i)}_{l_i}}e^{\mathrm{wt}(c)}
= \prod_{i=1}^LQ^{(k_i)}_{l_i}
\end{align}
due to (\ref{qf}).
In this sense $\chi_w(B)$ is a $w$-deformation of the character
$\mathrm{ch}\bigl(\otimes_{i=1}^L V(l_i \varpi_{k_i})\bigr)$.
See \cite{N} for a representation theoretical study.

\begin{example}\label{ex:lr}
From the description of the local energy in Section \ref{ss:kr}
and the Littlewood-Richardson rule (see Section \ref{ss:sle}) we have
\begin{align}
\chi_w(B^{(a)}_i \otimes B^{(1)}_s) &= 
\sum_{k=0}^{\min(i,s)} 
w^k s_{(i+s-k,i^{a-1},k)}(w_1,\ldots, w_{n+1}),
\label{k:ex1}\\
\chi_w(B^{(a)}_i \otimes B^{(r)}_1) &= 
\sum_{k=\max(0,r-a)}^{\min(r,n+1-a)} 
w^{a+k-\max(a,r)} s_{((i+1)^{r-k},i^{a-r+k},1^k)}(w_1,\ldots, w_{n+1}),
\end{align}
where $s_\lambda(w_1,\ldots, w_{n+1})$ is the Schur polynomial (\ref{k:sc}).
\end{example}

\subsection{Fermionic formula}
Given a tensor product
$B = B^{(k_1)}_{l_1} \otimes \cdots \otimes B^{(k_L)}_{l_L}$ and 
$\lambda \in P$, we define the Fermionic form\footnote{The 
$M(B,\lambda,w)$ here corresponds to 
\cite[eq.(4.3)]{HKOTY} with $q$ replaced by $w^{-1}$.} 
$M(B,\lambda,w) \in \Z_{\ge 0}[w]$ by
\begin{align}
M(B,\lambda,w) & = \sum_{m} w^{c(m)}
\prod_{1 \le a \le n, i\ge 1}
\left[{p^{(a)}_i +  m^{(a)}_i \atop m^{(a)}_i}\right]_w
\label{mdef}\\
p^{(a)}_i  &=  \sum_{j=1}^L \delta_{a,k_j}\min(i,l_j)
-\sum_{b=1}^nC_{ab}e^{(b)}_i,
\label{pai}\\
e^{(a)}_i &= \sum_{j \ge 1}\min(i,j)m^{(a)}_j,
\label{Eai}\\
c(m) &= \frac{1}{2}\sum_{1\le a, b \le n}C_{ab}
\sum_{i,j\ge 1}\min(i,j)m^{(a)}_i m^{(b)}_j,
\\
\left[{ m \atop k}\right]_w &= \frac{(w)_m}{(w)_k(w)_{m-k}},
\quad
(w)_m = \prod_{i=1}^m(1-w^{i}).
\end{align}
The quantity $p^{(a)}_i$ called vacancy and $e^{(a)}_i$ are the same with 
(\ref{ok:pai}) and (\ref{ok:eai}).
The sum $\sum_{m}$ in (\ref{mdef}) is taken over the array of 
nonnegative integers 
$m=(m^{(a)}_i)_{(a,i) \in [1,n]\times \Z_{\ge 1}}$
satisfying 
\begin{enumerate}
\item nonnegativity of vacancy\footnote{This finitely many conditions
are known to guarantee $p^{(a)}_i \ge 0$ 
for all $(a,i) \in [1,n]\times \Z_{\ge 1}$.}: 
$p^{(a)}_i \ge 0$ for all $(a,i) \in [1,n]\times \Z_{\ge 1}$
such that $m^{(a)}_i \ge 1$,

\item weight condition:
\begin{equation}\label{wc}
\sum_{a=1}^n e^{(a)}_\infty\alpha_a = 
\sum_{i=1}^L l_i\varpi_{k_i} - \lambda.
\end{equation}
\end{enumerate}
By the definition the summand corresponding to 
$m=(m^{(a)}_i)_{(a,i) \in [1,n]\times \Z_{\ge 1}}$ 
in (\ref{mdef}) is zero 
unless 
\begin{equation}\label{lze}
\lambda \in 
\Bigl(\sum_{i=1}^L l_i\varpi_{k_i} - Q_+\Bigr) \cap  P_+.
\end{equation}
The necessity for $\lambda \in P_+$ 
is seen by noting that (\ref{wc}) and (\ref{pai}) imply 
$\lambda = \sum_{a=1}^n p^{(a)}_\infty \varpi_a$. 
Given the data 
$(k_1,l_1), \ldots, (k_L,l_L) \in [1,n]\times \Z_{\ge 1}$,
there are finitely many choices of 
$m=(m^{(a)}_i)_{(a,i) \in [1,n]\times \Z_{\ge 1}}$
such that the above condition (i) and (ii) are satisfied for some $\lambda$
obeying (\ref{lze}). 
Those $m$ are called {\em configurations}. 
A configuration is equivalent to an $n$-tuple of 
Young diagrams via (\ref{mimi}).
They obey nontrivial constrains originating from the above (i) and (ii).  
To determine their asymptotic shape in the large $L$ limit 
is a main theme of this paper.

\begin{theorem} [\cite{KSS}]\label{th:kss}
For any $B = B^{(k_1)}_{l_1} \otimes \cdots \otimes B^{(k_L)}_{l_L}$,
the following equality is valid:
\begin{align}\label{k:xm}
\chi_w(B) = \sum_{\lambda} M(B,\lambda, w) \,\mathrm{ch}V(\lambda),
\end{align}
where the sum extends over those $\lambda$ satisfying (\ref{lze}). 
\end{theorem}

Theorem \ref{th:kss} tells that the Fermionic form 
$M(B,\lambda, w)$ is a 
$w$-analogue of the branching coefficient 
$[\otimes_{i=1}^L V(l_i\varpi_{k_i}): V(\lambda)]$.
The simplest case $L=1$ of Theorem \ref{th:kss} gives (\ref{qf}).
Namely one has
\begin{align}
Q^{(k)}_l = \chi_w(B^{(k)}_l),
\end{align}
which is actually independent of $w$\footnote{For types other than $A_n$,
$\chi_w(B^{(k)}_l)$ depends on $w$ in general.
Many such examples are available in \cite[app.A]{HKOTY}.}.
Fermionic forms for general affine Lie algebra were introduced 
for non-twisted  \cite{HKOTY} and twisted \cite{HKOTT} cases
inspired by those for the  
Kostka-Foulkes polynomials \cite{Ma} which correspond to 
$A^{(1)}_n$ \cite{KKR, KR1}. 

\subsection{Properties of deformed character}

\begin{proposition}[Th.6.1 in \cite{HKOTY}]\label{pr:mq}
Let $B = B^{(k_1)}_{l_1} \otimes \cdots \otimes B^{(k_L)}_{l_L}$ be arbitrary.
For any $(a,i) \in [1,n] \times \Z_{\ge 1}$
the following equality holds\footnote{When $i=1$,
the factor $B^{(a)}_0$ should just be dropped.}:
\begin{equation}\label{hko}
\chi_w(B^{(a)}_i \otimes B^{(a)}_i \otimes B)
= \chi_w(B^{(a)}_{i-1}\otimes B^{(a)}_{i+1} \otimes B)
+ w^\phi
\chi_w(\bigotimes_{b \sim a} B^{(b)}_i \otimes B),
\end{equation}
where $\phi = i + \sum_{j=1}^L \delta_{a, k_j}\min(i,l_j)$.
\end{proposition}

Actually (\ref{hko}) was shown in \cite{HKOTY} by substituting (\ref{k:xm})
to the three terms and using a decomposition of the Fermionic form.
As a corollary of Proposition \ref{pr:mq} and (\ref{k:kq}) with empty $B$
we see that 
the classical character $Q^{(k)}_l$ (\ref{qf}) of the KR module $W^{(k)}_l$ 
satisfies the Q-system:
\begin{align}\label{qsys}
(Q^{(a)}_i)^2 = Q^{(a)}_{i-1}Q^{(a)}_{i+1}
+ \prod_{b \sim a}Q^{(b)}_i.
\end{align}
To validate this at $i=0$ with $Q^{(a)}_0=1$, we employ the convention
$Q^{(a)}_{-1}=0$.

For simplicity we use the abbreviation 
\begin{align}\label{abb}
B_i = B^{(k_i)}_{l_i},\quad
Q_i = Q^{(k_i)}_{l_i}
\end{align}
in the remainder of this section and (\ref{k:rw}).
The following result resembles the Wick theorem.
\begin{lemma}\label{le:wick}
\begin{equation}\label{wick}
\frac{\partial}{\partial w} \log \chi_w(B_1\otimes \cdots \otimes B_L)|_{w=1}
= \sum_{1 \le i<j \le L}\frac{1}{Q_i Q_j}
\sum_{b \otimes c\in B_i \otimes B_j}
H(b \otimes c) e^{\mathrm{wt}(b \otimes c)}.
\end{equation}
\end{lemma}
The equality is {\em invalid} without specialization to $w=1$.  
\begin{proof}
From (\ref{k:xw}) and (\ref{k:d}) we have
\begin{align*}
\frac{\partial \chi_w(B)}{\partial w}\Big|_{w=1}&= 
\sum_{b_1 \otimes \cdots \otimes b_L}
D(b_1 \otimes \cdots \otimes b_L)
e^{\mathrm{wt}(b_1 \otimes \cdots \otimes b_L)}
=\sum_{1 \le i < j \le L}
\sum_{d}
H(b_i \otimes b^{(i+1)}_j)
e^{\mathrm{wt}(d)},
\end{align*}
where the sum is taken over 
$d =b_1\otimes \cdots \otimes b_i \otimes b^{(i+1)}_j
\otimes b'_{i+1} \otimes \cdots \otimes b'_{j-1} \otimes
b_{j+1} \otimes \cdots \otimes b_L
\in B_1 \otimes \cdots \otimes B_i \otimes 
B_{j} \otimes B_{i+1} \otimes \cdots \otimes B_{j-1}
\otimes B_{j+1} \otimes \cdots \otimes B_L$ in the notation of 
(\ref{brj}). 
By changing the summation variables, the summand corresponding to 
the pair $i<j$ is expressed as
\begin{align*}
\Biggl(\sum_{b_1 \otimes \cdots \otimes b_{i-1}}
\!\!\!
e^{\mathrm{wt}(b_1 \otimes \cdots \otimes b_{i-1})}\Biggr)
\Biggl(\sum_{b \otimes c \in B_i \otimes B_j}H(b \otimes c)
e^{\mathrm{wt}(b \otimes c)}\Biggr)
\Biggl(\sum_{b_{i+1} \otimes \cdots 
\otimes \overset{\vee}{b_j} \otimes  \cdots \otimes b_L}
\!\!\!\!\!\!\!\!\!\!\!\!
e^{\mathrm{wt}(b_{i+1} \otimes \cdots 
\otimes \overset{\vee}{b_j} \otimes  \cdots \otimes b_L)}\Biggr),
\end{align*}
where $\overset{\vee}{b_j} $ means the absence of the factor.
This equals $(\prod_{k=1,k\neq i,j}^LQ_k)
\sum_{b \otimes c\in B_i \otimes B_j}
H(b \otimes c) e^{\mathrm{wt}(b \otimes c)}$.
Therefore (\ref{wick}) follows from (\ref{k:kq}).
\end{proof}

\section{TBA analysis}\label{sec:tba}

\subsection{\mathversion{bold}Ensemble of $n$-tuple of Young diagrams
associated with Fermionic form}\label{ss:ff}

Having prepared the Fermionic form, we are ready to perform an asymptotic analysis of the randomized BBS on $B = (B^{(r)}_s)^{\otimes L}$ in Section \ref{sec:rbbs}. According to Proposition \ref{pr:st}, each local state obeys the probability distribution $\pi^{(r)}_s$ (\ref{noc}) on $B^{(r)}_s$. In other words a local state $u \in B^{(r)}_s$ occurs with the probability proportional to 
$e^{\mathrm{wt}(u)} = \prod_{a=1}^n z_a^{u_a}$
for $\mathrm{wt}(u) = u_1\varpi_1+\cdots + u_n \varpi_n \in P$. We shall concentrate on the regime of the parameters 
$z_1, \ldots, z_n$ such that
$z_a >0$ and $\prod_{b=1}^nz_b^{C_{ab}}>1$ for all $a \in [1,n]$. In view of $\prod_{b=1}^nz_b^{C_{ab}} = 
e^{\sum_{b=1}^nC_{ab}\varpi_b}= e^{\alpha_a}$,
it means $\alpha_a>0$ for all the simple roots $\alpha_a$ of $A_n$. Thus the local states closer to the highest weight element $u^{(r)}_s$ are realized with strictly larger probability. For instance, in case of $B_{s}^{(r)}=B_{1}^{(1)}$, this is equivalent to $\pi_{1}^{(1)}(1)>\pi_{1}^{(1)}(2)>\cdots>\pi_{1}^{(1)}(n+1)$.

Let $\mathcal{E}_L(r,s)$ be
the ensemble of states from $(B^{(r)}_s)^{\otimes L}$ 
generated by i.i.d.~probability distribution $(\pi^{(r)}_s)^{\otimes L}$.
Let further $\mathcal{E}^+_L(r,s)$ be the ensemble of highest states
$b_1 \otimes \cdots \otimes b_L \in ( B^{(r)}_s)^{\otimes L}$
whose probability distribution is proportional to 
$e^{\mathrm{wt}(b_1 \otimes \cdots \otimes b_L)}$.
The randomized BBS in Section \ref{sec:rbbs}
corresponds to $\mathcal{E}_L(r,s)$. 
It slightly differs from 
$\mathcal{E}^+_L(r,s)$ in which the local states are not i.i.d.~due 
to the nonlocal constraint of being highest.
Both of them induce a probability distribution on the 
set of $n$-tuple of Young diagrams 
$\mu_1, \ldots, \mu_n$ by taking the conserved quantities.
In the regime $\alpha_1, \ldots, \alpha_n >0$ under consideration,
the highest condition on $b_1 \otimes \cdots 
\otimes b_j \otimes b_{j+1} \otimes \cdots \otimes b_L 
\in (B^{(r)}_s)^{\otimes L}$ 
becomes void almost surely for the right part
$b_{j+1} \otimes \cdots \otimes b_L$ in the 
limit $L \gg j \rightarrow \infty$.
Since the large $L$ asymptotics $\mu_1, \ldots, \mu_n$ 
does not depend on the left finite tail of $b_1 \otimes \cdots \otimes b_L$,
we claim that those induced from 
$\mathcal{E}_L(r,s)$ and $\mathcal{E}^+_L(r,s)$ coincide.
(This ``asymptotic equivalence" of 
$\mathcal{E}_L(r,s)$ and $\mathcal{E}^+_L(r,s)$ is 
discussed in more detail for $B^{(r)}_s = B^{(1)}_1$ in \cite{KL}.)

For $\mathcal{E}^+_L(r,s)$, the conserved quantities 
$\mu_1, \ldots, \mu_n$ are the Young diagrams in the 
rigged configurations obtained by the KSS bijection.
Therefore their (joint) probability distribution is explicitly given by 
\begin{align}
\mathrm{Prob}(\mu_1,\ldots, \mu_n) 
&= \frac{1}{Z_L}
e^{-\sum_{a=1}^n \beta_a\sum_{i\ge 1}i m^{(a)}_i}
\prod_{1 \le a \le n, i\ge 1}
\left({p^{(a)}_i +  m^{(a)}_i \atop m^{(a)}_i}\right),
\label{prom}\\
m^{(a)}_i &= \text{$\#$ of length $i$ rows in $\mu_a$}.
\label{mimi}
\end{align}
In what follows we will identity 
the $n$-tuple of Young diagrams
$\mu=(\mu_1,\ldots, \mu_n)$ with the data
$m=(m^{(a)}_i)_{(a,i) \in [1,n]\times \Z_{\ge 1}}$
by (\ref{mimi}).
Then $|\mu_a| =\sum_{i \ge 1}im^{(a)}_i = e^{(a)}_\infty$ holds
from (\ref{Eai}).

The product of the binomial coefficients in (\ref{prom}) 
is the one in the Fermionic form (\ref{mdef}) at $w=1$.
It accounts for the multiplicity of 
$\mu=(\mu_1,\ldots, \mu_n)$ in the image of the KSS bijection.
See the explanation around (\ref{k:be}).
The vacancy $p^{(a)}_i$ is (\ref{pai}) with the choice
$\forall (k_i,l_i)=(r,s)$:
\begin{align}\label{ppai}
p^{(a)}_i = L\delta_{a,r}\min(i,s) - 
\sum_{b=1}^n C_{ab}e^{(b)}_i.
\end{align}
This will serve as the source of $(L,r,s)$-dependence of 
$\mathrm{Prob}(\mu_1,\ldots, \mu_n)$.
The parameters $\beta_1,\ldots, \beta_n$
are chemical potentials or inverse temperatures 
in the context of the generalized Gibbs ensemble. 
As we will see in (\ref{mz}),
they are actually the simple roots 
$\alpha_1, \ldots, \alpha_n$. 
Therefore the factor 
$e^{-\sum_{a=1}^n \beta_a\sum_{i\ge 1}i m^{(a)}_i}$ 
in (\ref{prom}) is just $e^{-Ls\varpi_r+\lambda}$ 
due to (\ref{wc}).
Besides the irrelevant constant $e^{-Ls\varpi_r}$,
the factor $e^\lambda$ here indeed incorporates the relative probability  
$e^{\mathrm{wt}(b_1 \otimes \cdots \otimes b_L)}$ adopted 
in $\mathcal{E}^+_L(r,s)$. 
Note that $\mathrm{Prob}(\mu_1,\ldots, \mu_n) =0$ unless 
$m=(m^{(a)}_i)_{(a,i) \in [1,n]\times \Z_{\ge 1}}$ is a configuration
in the sense explained after (\ref{lze}).
Finally $Z_L$ is given by (\ref{gcp}).

Our aim is to determine the ``equilibrium", i.e., most probable 
configuration under the probability distribution (\ref{prom})
when $L$ tends to infinity.
It will be done by the method of grand canonical ensemble 
with the partition function, namely the generating function of  (\ref{prom}):
\begin{align}
Z_L &=\sum_{m}
e^{-\sum^n_{a=1}e^{(a)}_\infty \beta_a}
\prod_{1 \le a \le n, i\ge 1}
\left({p^{(a)}_i +  m^{(a)}_i \atop m^{(a)}_i}\right)
\label{gcp}\\
&=e^{-Ls \varpi_r}  \sum_{\lambda}
M((W^{(r)}_s)^{\otimes L},\lambda, 1) \,e^\lambda.
\label{gcp2}
\end{align}
The latter expression tells that $Z_L$ is 
a generating series of the branching coefficients 
$[V(s\varpi_r) ^{\otimes L}: V(\lambda)]$.
Numerous combinatorial objects labeling the irreducible 
components $V(\lambda)$ and their counting formulas are known in 
combinatorial representation theory and algebraic combinatorics.

In the original work by Bethe himself \cite{Be},
a considerable effort was devoted to the completeness issue 
of his own {\em string hypothesis}.
The succeeding development \cite{KR1, KR2} 
assembled the Bethe strings and 
visualized them as rigged configurations.
These works produced the Fermionic counting formula (\ref{gcp}) for 
the representation theoretical quantity (\ref{gcp2}).
A further insight, soliton/string correspondence (see Section \ref{ss:eca}) 
gained after entering this century, 
elucidated that the Bethe strings are nothing but the BBS {\em solitons} 
for which one can formulate an integrable dynamics based on KR crystals. 
It endowed the {\em individual} term in the sum (\ref{gcp}) with a 
natural interpretation as the {\em partition function} 
of the BBS with a prescribed soliton content $m$ \cite{KOTY,KOSTY}.
In short the BBS provided the Fermionic formula with a {\em refinement} 
via a {\em quasi-particle picture}.
Physically speaking the BBS solitons are bound states of magnons 
over a ferromagnetic ground state of an integrable $A^{(1)}_n$-symmetric 
spin chain deformed by $U_{q=0}(A^{(1)}_n)$.  

\subsection{TBA equation and Y-system}
We are going to 
apply the idea of TBA \cite{YY} 
to the system governed by the grand canonical partition function (\ref{gcp}).
Similar problems have been studied in the context of 
ideal gas of Haldane exclusion statistics.
See for example the original works \cite{Su2, Wu, IA} 
and a review from the viewpoint of a generalized Q-system 
\cite[sec.13]{KNS}.
In fact our treatment here is a constant (spectral parameter free)
version of the TBA analysis in \cite[sec.14, 15]{KNS}.
In Theorem \ref{thk:main} it will be shown that 
the results coincide quite nontrivially with 
those obtained from the crystal theory consideration.

In the large $L$ limit, 
the dominant contribution in (\ref{gcp}) come from 
those $m=(m^{(a)}_i)$ exhibiting the 
$L$-linear asymptotic behavior
\begin{align}\label{mpel}
m^{(a)}_i \simeq L \rho^{(a)}_i,\quad
p^{(a)}_i \simeq L \sigma^{(a)}_i,\quad
e^{(a)}_i \simeq L \varepsilon^{(a)}_i,\quad
|\mu_a| = e^{(a)}_\infty \simeq L \varepsilon^{(a)}_\infty,
\end{align}
where $\rho^{(a)}_i,\sigma^{(a)}_i, \varepsilon^{(a)}_i$ are of $O(L^0)$.
This fact will be justified by invoking the 
large deviation principle in \cite{KL}.
From (\ref{ppai}) and (\ref{Eai}) the scaled variables are related as 
\begin{align}\label{ser}
&\sigma^{(a)}_i = \delta_{a,r}\min(i,s) 
- \sum_{b=1}^nC_{ab}\,\varepsilon^{(b)}_i,
\quad
\varepsilon^{(a)}_i= \sum_{j\ge 1}\min(i,j)\rho^{(a)}_j.
\end{align}
This is a constant version of 
the Bethe equation in terms of string density $\rho^{(a)}_i$ and 
the hole density $\sigma^{(a)}_i$ in 
\cite[eq.(15.6)]{KNS}$|_{l \rightarrow \infty}$.
The equilibrium configuration corresponds to
the $\rho=(\rho^{(a)}_i)$ that minimizes the ``free energy per site"
\begin{align}\label{Fn}
F[\rho] = \sum_{a=1}^n\beta_a\sum_{i=1}^l i\rho^{(a)}_i 
-
\sum_{a=1}^n\sum_{i=1}^l 
\Bigl((\rho^{(a)}_i + \sigma^{(a)}_i)\log(\rho^{(a)}_i + \sigma^{(a)}_i)
- \rho^{(a)}_i \log \rho^{(a)}_i -\sigma^{(a)}_i \log \sigma^{(a)}_i\Bigr).
\end{align}
This is $(-1/L)$ times logarithm of the summand in (\ref{gcp}) 
to which the Stirling formula has been applied.
Note that (\ref{mpel}) is 
consistent with the extensive property of the free energy, which enabled us 
to remove the system size $L$ as a common overall factor.
We have introduced a cut-off $l$ for the index $i$,
which will be sent to infinity later.
Accordingly the latter relation 
in (\ref{ser}) should  be understood as 
$\varepsilon^{(a)}_i= \sum_{j=1}^l\min(i,j)\rho^{(a)}_j$.
From $\frac{\partial \sigma^{(b)}_j}{\partial \rho^{(a)}_i} = -C_{ab}\min(i,j)$, 
one finds that the equilibrium 
condition $\frac{\partial F[\rho]}{\partial \rho^{(a)}_i} = 0$
is expressed as a TBA equation
\begin{align}
-i\beta_a + \log(1+Y^{(a)}_i)=
\sum_{b=1}^nC_{ab}\sum_{j=1}^l\min(i,j)
\log(1+(Y^{(b)}_j)^{-1})
\label{TBAn}
\end{align}
for $1 \le i \le l$ in terms of the ratio
\begin{align}\label{yrs}
Y^{(a)}_i = \frac{\sigma^{(a)}_i}{\rho^{(a)}_i}. 
\end{align}
The TBA equation is equivalent to the Y-system
\begin{align}\label{YY}
\frac{(1+Y^{(a)}_{i})^2}{
(1+Y^{(a)}_{i-1})(1+Y^{(a)}_{i+1})}
= \prod_{b=1}^n
(1+(Y^{(b)}_i)^{-1})^{C_{ab}}
\end{align}
for $1 \le i \le l$ 
with the boundary condition
\begin{align}\label{ybd}
Y^{(a)}_0 =0, \quad 1+Y^{(a)}_{l+1} = e^{\beta_a}(1+Y^{(a)}_l).
\end{align}
The Y-system is known to follow from the 
Q-system (\ref{qsys}) by the substitution (cf. \cite[Prop. 14.1]{KNS})
\begin{align}\label{yqw}
Y^{(a)}_i = \frac{Q^{(a)}_{i-1}Q^{(a)}_{i+1}}
{\prod_{b\sim a} Q^{(b)}_i},
\qquad 
1+Y^{(a)}_i = \prod_{b=1}^n (Q^{(b)}_i)^{C_{ab}},
\qquad
1+(Y^{(a)}_i)^{-1}
= \frac{(Q^{(a)}_i)^2}{Q^{(a)}_{i-1}Q^{(a)}_{i+1}},
\end{align}
where $Q^{(a)}_i \in \Z[z^{\pm 1}_1,\ldots, z^{\pm 1}_n]$ 
is defined in (\ref{qf}).
Now we take the boundary condition (\ref{ybd}) into account.
The left one $Y^{(a)}_0=0$  
is automatically satisfied due to $Q^{(a)}_{-1}=0$.
See the remark after (\ref{qsys}).
On the other hand the right condition in (\ref{ybd}) is expressed as
\begin{align}\label{qoq}
e^{\beta_a} = \prod_{b=1}^n\left(
\frac{Q^{(b)}_{l+1}}{Q^{(b)}_l}\right)^{C_{ab}}.
\end{align}
The result 
\cite[Th. 7.1 (C)]{HKOTY} tells that 
$\lim_{l \rightarrow \infty}
(Q^{(a)}_{l+1}/Q^{(a)}_l)= e^{\varpi_a}$
in the regime $\prod_{b=1}^nz_b^{C_{ab}}>1$ under consideration.
Thus the large $l$ limit of (\ref{qoq}) can be taken, giving 
\begin{align}\label{mz}
e^{\beta_a} = \prod_{b=1}^n z^{C_{ab}}_b = e^{\alpha_a}.
\end{align}
In this way the chemical potentials $\beta_a$ 
are naturally identified with the simple roots $\alpha_a$.
We shall keep using the both symbols although.

To summarize so far, 
we have determined the equilibrium configuration  
$\rho_{\mathrm{eq}}$ of $\rho=(\rho^{(a)}_i)$
implicitly by (\ref{ser}), (\ref{yrs}), (\ref{yqw}) and (\ref{mz})
in terms of the chemical potentials $\beta_1, \ldots, \beta_n$.
The next task is to 
relate them to the canonically conjugate densities which are 
``physically more controllable".
It amounts to formulating the equation of state.
This we do in the next subsection.

\subsection{Equation of state for randomized BBS}

From now on we will only treat the equilibrium values and 
frequently omit mentioning it.
Let us calculate the equilibrium value of the free energy per site (\ref{Fn}).
First we use (\ref{yrs}) to rewrite (\ref{Fn}) as
\begin{align}\label{F1}
F[\rho_{\mathrm{eq}}] = \sum_{a=1}^n\beta_a\sum_{i=1}^l i\rho^{(a)}_i 
-\sum_{a=1}^n\sum_{i=1}^l 
\Bigl(\rho^{(a)}_i\log(1+Y^{(a)}_i) 
+ \sigma^{(a)}_i\log(1+(Y^{(a)}_i)^{-1})\Bigr).
\end{align}
On the other hand taking the linear combination of the TBA equation as
$\sum_{a=1}^n\sum_{i=1}^l (\ref{TBAn}) \times \rho^{(a)}_i$ we get
\begin{align}
\sum_{a=1}^n\beta_a\sum_{i=1}^li\rho^{(a)}_i =
\sum_{a=1}^n\sum_{i=1}^l\rho^{(a)}_i\log(1+Y^{(a)}_i)
-\sum_{a,b=1}^n C_{ab}\sum_{i,j=1}^l\min(i,j)
\rho^{(a)}_i\log(1+(Y^{(b)}_j)^{-1}).
\end{align}
Substituting this into the first term on the RHS of (\ref{F1}) and 
using $\sigma^{(a)}_i$ from (\ref{ser}) we find
\begin{align}\label{Fm}
F[\rho_{\mathrm{eq}}] = -\sum_{i=1}^l
\min(i,s)\log(1+(Y^{(r)}_i)^{-1}) = 
-\log\left(Q^{(r)}_s\Bigl(\frac{Q^{(r)}_l}{Q^{(r)}_{l+1}}\Bigr)^s\right)
\overset{l \rightarrow \infty}{\longrightarrow}
-\log\left(z^{-s}_rQ^{(r)}_s\right),
\end{align}
where (\ref{yqw}) is used and $l \gg s$ is assumed in the second equality.

Now we resort to the general relation
\begin{align}\label{FZ}
F[\rho_{\mathrm{eq}}]
= -\lim_{L \rightarrow \infty}\frac{1}{L}\log Z_L.
\end{align}
From (\ref{mpel}), 
we see that the equilibrium (most probable) value of the 
$1/L$-scaled number of boxes $L^{-1}|\mu_1|, \ldots, L^{-1}|\mu_n|$ 
in the Young diagrams $\mu_1,\ldots, \mu_n$ are that for 
$\varepsilon^{(1)}_\infty, \ldots, \varepsilon^{(n)}_\infty$.
Denote them by $\nu_1, \ldots, \nu_n$.
These are parameters representing the {\em densities} 
of the boxes in the Young diagrams $\mu_1, \ldots, \mu_n$.
From (\ref{gcp}), (\ref{mpel})  and (\ref{FZ}) one has the relation 
$\nu_a
= \frac{\partial F[\rho_\mathrm{eq}]}{\partial \beta_a}$ for 
$1 \le a \le n$.
In view of (\ref{mz}) 
it is convenient to take the linear combination of this as follows:
\begin{align}
\sum_{b=1}^nC_{ab}\nu_b = 
\left(\sum_{b=1}^nC_{ab}\frac{\partial}{\partial \beta_b}\right)
F[\rho_\mathrm{eq}]
= z_a\frac{\partial F[\rho_\mathrm{eq}]}{\partial z_a}.
\end{align}
Substituting (\ref{Fm}) we arrive at 
the equation of state of the system:
\begin{align}\label{est}
z_a\frac{\partial}{\partial z_a}\log Q^{(r)}_s 
= s\delta_{a,r}-\sum_{b=1}^nC_{ab}\,\nu_b
\qquad (1 \le a \le n).
\end{align}
The LHS is an explicit rational function 
of $z_1,\ldots, z_n$ that can be calculated from (\ref{k:sc}) and (\ref{qf}).
The variables $z_1,\ldots, z_n$ are simply related to the chemical potentials 
$\beta_1, \ldots, \beta_n$ 
or equivalently 
to the fugacity $e^{-\beta_1}, \ldots, e^{-\beta_n}$ by (\ref{mz})
as
\begin{align}\label{k:hro}
z_a = e^{\sum_{b=1}^n(C^{-1})_{ab}\beta_b}
= e^{\sum_{b=1}^n \bigl(\min(a,b)-\frac{ab}{n+1}\bigr)\beta_b}.
\end{align}
Thus (\ref{est}) relates the densities 
$\nu_1, \ldots, \nu_n$
with the fugacity $e^{-\beta_1}, \ldots, e^{-\beta_n}$, 
thereby enabling us to control either one by the other. 
Set
\begin{align}\label{a:qb}
y_a = e^{-\alpha_a},\qquad
\overline{Q}^{(r)}_s= z^{-s}_r Q^{(r)}_s \in 
\Z[y_1,\ldots, y_n],
\end{align}
where $y_1, \ldots, y_n$ are the fugacity mentioned just above since  
$\alpha_a=\beta_a$ according to (\ref{mz}).
Then the equation of state (\ref{est}) 
also admits a somewhat simpler presentation as
\begin{align}\label{k:mhn}
\nu_a = y_a \frac{\partial}{\partial y_a}\log \overline{Q}^{(r)}_s.
\end{align}
To see this, note 
from $y_a = \prod_{b=1}^n z_b^{-C_{ab}}$ and (\ref{a:qb}) that 
(\ref{est}) is rewritten as
\begin{align}
\sum_{b=1}^nC_{ab}\nu_b = 
s\delta_{a,r}-z_a\frac{\partial}{\partial z_a}\log\bigl(
z^s_r \overline{Q}^{(r)}_s \bigr) 
= -z_a\frac{\partial}{\partial z_a}\log
\overline{Q}^{(r)}_s = 
\sum_{b=1}^n C_{ab}\,
y_b \frac{\partial}{\partial y_b}\log \overline{Q}^{(r)}_s.
\end{align}

In the language of the BBS,
the relation $\varepsilon^{(a)}_\infty = \nu_a$ implies that 
the equilibrium weight of a site variable $b \in B^{(r)}_s$ is
\begin{align}\label{k:ew}
\mathrm{wt}(b) = s\varpi_r - \sum_{a=1}^n \nu_a \alpha_a\;\; \leftrightarrow \;\;
\#_a(b) = s\theta(a \le r) +\nu_{a-1}-\nu_a\quad(\nu_0=\nu_{n+1}=0),
\end{align}
where $\#_a(b)$ for $a \in [1,n+1]$ denotes the number of the letter $a$ 
in $b \in B^{(r)}_s$ regarded as a semistandard tableau of shape $(s^r)$.
The empty space corresponds to the letter $1$.
Note that the weight $\mathrm{wt}(b)$ specifies 
an element $b \in B^{(r)}_s$ uniquely if and only if 
$\min(r,n-r,s)=1$.

\subsection{Difference equation characterizing the equilibrium shape}\label{ss:de}

One should recognize (\ref{est}) as 
the $i=\infty$ case of the left relation in (\ref{ser}).
It concerns the total number (density) of boxes in the Young diagrams only.
However, the relations (\ref{ser}) and (\ref{yrs}) 
hold for any finite $i$ and provide information on the equilibrium {\em shape}
of these Young diagrams.
In fact they can be combined to give the difference equation
for the variables $\varepsilon^{(a)}_i$ as
\begin{align}\label{con}
\delta_{a,r}\min(i,s) - \sum_{b=1}^nC_{ab}\varepsilon^{(b)}_i 
= Y^{(a)}_i(-\varepsilon^{(a)}_{i-1}
+2\varepsilon^{(a)}_i - \varepsilon^{(a)}_{i+1}).
\end{align} 
The quantity in the 
parenthesis in the RHS is $\rho^{(a)}_i$ (\ref{ser}), which is the 
$1/L$-scaled number $m^{(a)}_i$ of length $i$ rows in 
the Young diagram $\mu_a$.
See (\ref{mimi}) and (\ref{mpel}).

In this way we have characterized the vertically $1/L$-scaled 
equilibrium shape of the Young diagrams $\mu_1, \ldots, \mu_n$
under the prescribed densities 
$(\nu_1, \ldots, \nu_n)= \lim_{L \rightarrow \infty}L^{-1}
(|\mu_1|, \ldots, |\mu_n|)$
in terms of the variables $\varepsilon^{(a)}_i$.
The procedure consists of the following steps:

\begin{enumerate}

\item 
Given the densities $\nu_1, \ldots, \nu_n$,
determine $z_1, \ldots, z_n$ by the equation of state (\ref{est}) or 
(\ref{k:mhn}).

\item Compute $Y^{(a)}_i$ by substituting those $z_1, \ldots, z_n$ into 
(\ref{yqw}) and (\ref{qf}).

\item Find the solution to (\ref{con}) with the 
boundary condition 
$\varepsilon^{(a)}_0=0, \varepsilon^{(a)}_\infty=\nu_a$.

\end{enumerate} 

Once $\varepsilon^{(a)}_i$ is known, 
the equilibrium Young diagrams 
are deduced from
\begin{align}\label{k:es5}
\varepsilon^{(a)}_i = 
\lim_{L \rightarrow \infty}
\frac{1}{L}\bigl(\text{$\#$ of boxes in the left $i$ columns of 
$\mu_a$}\bigr). 
\end{align}
See (\ref{mimi}), (\ref{mpel}) and (\ref{ser}).
In the next subsection we present an explicit solution to the step (iii)
in terms of the stationary local energy (\ref{k:ep}).

\subsection{Solution to the difference equation by stationary local energy}

From (\ref{noc}) the stationary local energy (\ref{k:ep}) 
(for the i.i.d.~ensemble $\mathcal{E}_L(r,s)$ with $L \rightarrow \infty$) 
is expressed as 
\begin{align}\label{k:izm}
h^{(a)}_i = \frac{\sum_{x \otimes y
\in B^{(a)}_i \otimes B^{(r)}_s}H(x \otimes y)
e^{\mathrm{wt}(x \otimes y)}}
{Q^{(a)}_i Q^{(r)}_s}.
\end{align}
According to (\ref{k:hai}) this is equal to the 
$\lim_{L \rightarrow \infty}\frac{1}{L}E^{(a)}_i$. 
On the other hand (for the highest state ensemble
$\mathcal{E}^+_L(r,s)$ with $L \rightarrow \infty$),
the soliton/string correspondence (\ref{k:ii}) and 
the definition (\ref{mpel}) indicate that 
the same quantity should also show up in the TBA analysis exactly 
as $\varepsilon^{(a)}_i$.
Thus the asymptotic equivalence of the two ensembles indicates that 
they coincide.
The next theorem, which is our main result in the paper,
identifies them rigorously.
\begin{theorem}\label{thk:main}
The solution to the difference equation (\ref{con})
satisfying the boundary condition 
$\varepsilon^{(a)}_0=0, \varepsilon^{(a)}_\infty=\nu_a$
is provided by the stationary local energy
$\varepsilon^{(a)}_i=h^{(a)}_i$ in (\ref{k:izm}).
\end{theorem}
Theorem \ref{thk:main} shows a nontrivial coincidence of the
equilibrium configuration, i.e., the $n$-tuple of Young diagrams,  
determined by two different approaches:

\begin{itemize}

\item
stationary local energy for 
the Markov process of carriers in the randomized BBS,

\item
difference equation arising from the 
TBA analysis of the Fermionic formula.

\end{itemize}
The result may be regarded as 
{\em randomized} version of the 
soliton/string correspondence (\ref{k:ii}).
Being able to give an explicit formula for $\varepsilon^{(a)}_i$ 
is a very rare event   
in the actual TBA analyses involving the spectral parameter.

For simplicity we temporarily write 
the $w$-deformed character (\ref{k:xw}) as
\begin{align}\label{k:sn}
Q^{(k_1)}_{l_1} \ast \cdots \ast Q^{(k_L)}_{l_L}
= \chi_w(B^{(k_1)}_{l_1} \otimes \cdots 
\otimes B^{(k_L)}_{l_L}).
\end{align}
In this notation, Lemma \ref{le:wick} reads as
\begin{align}\label{k:rw}
\frac{\partial \log(Q_1\ast \cdots \ast Q_L)}{\partial w}
\Big|_{w=1} &= \sum_{1 \le i < j \le L}
\frac{\partial \log(Q_i \ast Q_j)}{\partial w} 
\Big|_{w=1}.
\end{align} 
Then Theorem \ref{thk:main} is summarized in the following formula
for the $1/L$-scaled Young diagrams:
\begin{align}\label{k:es}
\varepsilon^{(a)}_i = 
\frac{\partial \log(Q^{(a)}_i \ast Q^{(r)}_s)}{\partial w} 
\Big|_{w=1}.
\end{align}
From this and (\ref{k:es5}) the quantity 
$\eta^{(a)}_i := \varepsilon^{(a)}_i - \varepsilon^{(a)}_{i-1}$
has the meaning and the explicit formula as
\begin{equation}\label{k:eta}
\eta^{(a)}_i= \lim_{L \rightarrow \infty}
\frac{1}{L}\bigl(\text{$\#$ of boxes in the $i$-th column of 
$\mu_a$}\bigr)=
\frac{\partial}{\partial w} \Biggl(\log
\frac{Q^{(a)}_i \ast Q^{(r)}_s}
{Q^{(a)}_{i-1} \ast Q^{(r)}_s} \Biggr)
\Bigg|_{w=1}.
\end{equation}

\subsection{Proof of Theorem \ref{thk:main}}

First we prove
\begin{proposition}\label{pr:de}
$\varepsilon^{(a)}_i = h^{(a)}_i$ (\ref{k:izm}) 
provides a solution to
the difference equation (\ref{con}).
\end{proposition}
\begin{proof}
By substitution of  
(\ref{k:izm}) and the formula (\ref{yqw}) for $Y^{(a)}_i$, 
the equation (\ref{con}) becomes 
\begin{equation}\label{izm2}
\begin{split}
&\delta_{a,r}\min(i,s) - \sum_{b=1}^nC_{ab}
\frac{\sum H(x^{(b)}_i \otimes y)
e^{\mathrm{wt}(x^{(b)}_i \otimes y)}}
{Q^{(b)}_iQ^{(r)}_s} 
=\frac{Q^{(a)}_{i-1}Q^{(a)}_{i+1}}{Q^{(r)}_s
\prod_{b \sim a} Q^{(b)}_i}
\\
&\times
\left(
-\frac{\sum H(x^{(a)}_{i-1} \otimes y)
e^{\mathrm{wt}(x^{(a)}_{i-1}\otimes y)}}{Q^{(a)}_{i-1}}
+2\frac{\sum H(x^{(a)}_i \otimes y)
e^{\mathrm{wt}(x^{(a)}_i\otimes y)}}{Q^{(a)}_{i}}
-\frac{\sum H(x^{(a)}_{i+1} \otimes y)
e^{\mathrm{wt}(x^{(a)}_{i+1}\otimes y)}}{Q^{(a)}_{i+1}}
\right).
\end{split}
\end{equation}
Here and in what follows $x^{(c)}_j$ and $y$ should always be summed over 
$B^{(c)}_j$ and $B^{(r)}_s$, respectively.
By removing the denominators, this is cast into 
\begin{equation}\label{izm3}
\begin{split}
\delta_{a,r}\min(i,s)Q^{(r)}_s\prod_{b \sim a}Q^{(b)}_i
&=2 Q^{(a)}_i \sum H(x^{(a)}_i\otimes y)
e^{\mathrm{wt}(x^{(a)}_i \otimes y)}\\
&-\sum_{b\sim a}\Bigl(\prod_{c\sim a, c \neq b}Q^{(c)}_i\Bigr)
\sum H(x^{(b)}_i \otimes y)
e^{\mathrm{wt}(x^{(b)}_i \otimes y)}
\\
&- Q^{(a)}_{i+1}\sum H(x^{(a)}_{i-1} \otimes y) 
e^{\mathrm{wt}(x^{(a)}_{i-1}\otimes y)}
- Q^{(a)}_{i-1}\sum H(x^{(a)}_{i+1}\otimes y)
e^{\mathrm{wt}(x^{(a)}_{i+1}\otimes y)}.
\end{split}
\end{equation}
In the derivation, we have used the Q-system (\ref{qsys}) to 
cancel a factor $Q^{(a)}_i$ in the first term of the RHS.
In order to verify (\ref{izm3}) we consider the two special cases of (\ref{hko}):
\begin{align}
&Q^{(a)}_i \ast Q^{(a)}_i \ast Q^{(r)}_s 
- Q^{(a)}_{i-1}\ast Q^{(a)}_{i+1} \ast Q^{(r)}_s  
- w^\kappa \prod_{b \sim a}
Q^{(b)}_i\ast Q^{(r)}_s = 0,
\label{k:mk}\\
&
Q^{(a)}_i \ast Q^{(a)}_i 
- Q^{(a)}_{i-1} \ast Q^{(a)}_{i+1}
-w^i \prod_{b \sim a} Q^{(b)}_i = 0,
\label{k:mk2}
\end{align}
where $\kappa = i+\delta_{a,r}\min(i,s)$ and 
the product over $b$ means  the one by $\ast$.
Take the $w$-derivative of (\ref{k:mk}) at $w=1$.
By means of Lemma \ref{le:wick} or equivalently (\ref{k:rw}), it leads to
\begin{equation}\label{qr1}
\begin{split}
0 &= 2Q^{(a)}_i\sum H(x^{(a)}_i \otimes y) 
e^{\mathrm{wt}(x^{(a)}_i \otimes y)} 
- Q^{(a)}_{i+1}\sum H(x^{(a)}_{i-1} \otimes y)
 e^{\mathrm{wt}(x^{(a)}_{i-1} \otimes y)} \\
&- Q^{(a)}_{i-1}\sum H(x^{(a)}_{i+1} \otimes y)
 e^{\mathrm{wt}(x^{(a)}_{i+1} \otimes y)} 
-\kappa \Bigl(\prod_{b \sim a}Q^{(b)}_i\Bigr)  Q^{(r)}_s 
- \sum_{b \sim a}\Bigl(\prod_{c\sim a, c \neq b}Q^{(c)}_i\Bigr)
H(x^{(b)}_i \otimes y)
 e^{\mathrm{wt}(x^{(b)}_{i} \otimes y)} 
\\
&+\Biggl(
\frac{\partial (Q^{(a)}_i \ast Q^{(a)}_i)}{\partial w}
- \frac{\partial (Q^{(a)}_{i-1} \ast Q^{(a)}_{i+1})}{\partial w}
- \sum_{b,c \sim a}
\Bigl(\prod_{d\sim a, d \neq b,c}Q^{(d)}_i\Bigr)
\frac{\partial (Q^{(b)}_i \ast Q^{(c)}_i)}{\partial w}
\Biggr)\Big|_{w=1}Q^{(r)}_s.
\end{split}
\end{equation}
The same calculation for (\ref{k:mk2}) tells that 
the quantity in the big parenthesis of the 
last line of (\ref{qr1}) is equal 
to $i\prod_{b \sim  a}Q^{(b)}_i$ at $w=1$.
Therefore this term cancels the $\kappa$ term on the second line of (\ref{qr1}) partially.
The resulting relation is nothing but (\ref{izm3}).
\end{proof}

Next we verify the boundary condition
$h^{(a)}_0=0, h^{(a)}_\infty=\nu_a$.
As the former is obvious, we concentrate on the latter.
From (\ref{k:mhn}) and (\ref{k:izm}) 
the boundary condition $h^{(a)}_\infty=\nu_a$ in question is stated as
\begin{proposition}\label{pr:bcd}
\begin{align}\label{k:otm}
\lim_{i\rightarrow \infty}
\frac{\sum_{x \otimes y
\in B^{(a)}_i \otimes B^{(r)}_s}H(x \otimes y)
e^{\mathrm{wt}(x \otimes y)}}
{Q^{(a)}_i Q^{(r)}_s}
= y_a \frac{\partial}{\partial y_a}\log \overline{Q}^{(r)}_s.
\end{align}
\end{proposition}
Proposition \ref{pr:bcd} 
turns out to be reducible to some simple cases.
To demonstrate it we utilize the $(r,s)$-dependence of $h^{(a)}_i$ (\ref{k:izm}),
hence exhibit it as
\begin{align}\label{a:hde}
h^{(a,r)}_{i,s} = h^{(r,a)}_{s,i}
= \frac{\sum_{x \otimes y
\in B^{(a)}_i \otimes B^{(r)}_s}H(x \otimes y)
e^{\mathrm{wt}(x \otimes y)}}
{Q^{(a)}_i Q^{(r)}_s}.
\end{align}
Here the symmetry under the exchange of the indices is due to the
invariance of weights and the local energy $H$ 
by the combinatorial $R$. 

\begin{lemma}\label{le:erk}
\begin{align}\label{a:hrec}
s\delta_{a,r} = 2y_r^{-s}\prod_{t=1}^n
\bigl(\overline{Q}^{(t)}_s\bigr)^{C_{rt}}
h^{(a,r)}_{\infty, s}
- y^{-s}_r
\frac{\overline{Q}^{(r)}_{s+1}\overline{Q}^{(r)}_{s-1}}
{\prod_{t\sim r}\overline{Q}^{(t)}_s}
\bigl(h^{(a,r)}_{\infty, s-1}+h^{(a,r)}_{\infty, s+1}\bigr)
-\sum_{t\sim r}h^{(a,t)}_{\infty,s}.
\end{align}
\end{lemma}
\begin{proof}
In terms of $h^{(a,r)}_{i,s}$ in (\ref{a:hde}), 
the already established relation
(\ref{izm2}) or equivalently (\ref{izm3}) 
reads
\begin{align*}
\delta_{a,r}\min(i,s) = 2\prod_{b=1}^n\bigl(Q^{(b)}_i\bigr)^{C_{ab}}
h^{(a,r)}_{i,s}
-\frac{Q^{(a)}_{i+1}Q^{(a)}_{i-1}}{\prod_{b \sim a}Q^{(b)}_i}
(h^{(a,r)}_{i-1,s}+h^{(a,r)}_{i+1,s})
-\sum_{b \sim a} h^{(b,r)}_{i,s}.
\end{align*}
Exchange the indices $(a,i) \leftrightarrow (r,s)$ here and apply the 
symmetry $h^{(r,a)}_{s,i} = h^{(a,r)}_{i,s}$.
Then (\ref{a:hrec}) follows from it by taking the limit $i \rightarrow \infty$
and substituting $Q^{(r)}_s = z^s_r \overline{Q}^{(r)}_s$.
\end{proof}

\begin{lemma}\label{le:slk}
Let 
$\zeta^{(a,r)}_s 
= y_a \frac{\partial}{\partial y_a}\log \overline{Q}^{(r)}_s$
be the RHS of (\ref{k:otm}).
It satisfies
\begin{align}
s\delta_{a,r} = 2y_r^{-s}\prod_{t=1}^n
\bigl(\overline{Q}^{(t)}_s\bigr)^{C_{rt}}
\zeta^{(a,r)}_{s}
- y^{-s}_r
\frac{\overline{Q}^{(r)}_{s+1}\overline{Q}^{(r)}_{s-1}}
{\prod_{t\sim r}\overline{Q}^{(t)}_s}
\bigl(\zeta^{(a,r)}_{s-1}+\zeta^{(a,r)}_{s+1}\bigr)
-\sum_{t\sim r}\zeta^{(a,t)}_{s}.
\end{align}
\end{lemma}
\begin{proof}
The Q-system (\ref{qsys}) becomes
$(\overline{Q}^{(r)}_s)^2 = 
\overline{Q}^{(r)}_{s-1}\overline{Q}^{(r)}_{s+1}
+ y^s_r \prod_{t \sim r}\overline{Q}^{(t)}_s$
in terms of the variables in (\ref{a:qb}).
This is an identity in $\Z[y_1,\ldots, y_n]$.
The assertion follows from it by taking the derivative
$y_a \frac{\partial}{\partial y_a}$.
\end{proof}

\vspace{0.2cm}\noindent
{\em Proof of Proposition \ref{pr:bcd}}.
From Lemma \ref{le:erk} and Lemma \ref{le:slk}, 
the quantities $h^{(a,r)}_{\infty,s}$ and $\zeta^{(a,r)}_{s}$
obey the same difference relation with respect to $r$
which is at most of 
second order since $t \sim r$ means $t \in \{r \pm 1\}\cap [1,n]$.
Moreover they are both $0$ at $r=0$.
Therefore $h^{(a,r)}_{\infty,s}=\zeta^{(a,r)}_{s}$ (\ref{k:otm}) 
follows from the $r=1$ case.
It can be slightly rewritten by (\ref{a:qb}) as 
\begin{align}\label{k:mka}
\lim_{i\rightarrow \infty}
\frac{\sum_{x \otimes y
\in B^{(a)}_i \otimes B^{(1)}_s}H(x \otimes y)
e^{\mathrm{wt}(x \otimes y)}}{Q^{(a)}_i z_1^s}
= y_a \frac{\partial \overline{Q}^{(1)}_s}{\partial y_a}.
\end{align} 
In the sequel we prove (\ref{k:mka}).
From the $w$-derivative of (\ref{k:ex1}) at $w=1$, 
the LHS of (\ref{k:mka}) 
with fixed $i ( \ge s)$ 
is expressed as
\begin{align}\label{k:cdk}
\frac{\sum_{x \otimes y
\in B^{(a)}_i \otimes B^{(1)}_s}H(x \otimes y)
e^{\mathrm{wt}(x \otimes y)}}{Q^{(a)}_i z_1^s}
= \sum_{k=1}^s k
\frac{s_{(i+s-k,i^{a-1}, k)}(w_1,\ldots, w_{n+1})}
{s_{(i^{a})}(w_1,\ldots, w_{n+1}) \,z^{s}_1}.
\end{align}
In the regime $\prod_{b=1}^nz_b^{C_{ab}}>1$ under consideration,
the variables $w_a$ in (\ref{k:wde}) satisfy
$w_1> w_2 > \cdots > w_{n+1}$.
Then the large $i$ limit of each summand in RHS of (\ref{k:cdk}) is 
easily extracted from the determinantal formula in (\ref{k:sc}).
It is decomposed into a product of 
Schur polynomials, which leads to 
\begin{equation}\label{k:hs}
\begin{split}
\text{LHS of (\ref{k:mka})}
&= z^{-s}_1 \sum_{k=1}^s k\, s_{(s-k)}(w_1,\ldots, w_a)
s_{(k)}(w_{a+1},\ldots, w_n)\\
&=\sum_{k=0}^s k\,
s_{(s-k)}(\overline{w}_1,\ldots, \overline{w}_a)
s_{(k)}(\overline{w}_{a+1},\ldots, \overline{w}_n).
\end{split}
\end{equation}
We have set 
$\overline{w}_a=w_a/z_1 = y_1y_2\cdots y_{a-1}$.
See (\ref{k:wde}) and (\ref{a:qb}).
As for RHS of (\ref{k:mka}), 
we invoke the formula for $\overline{Q}^{(1)}_s$ 
as the sum over semistandard tableaux on the Young diagram
with length $s$ single row shape.
The entry $b \in [1,n+1]$ of the tableaux corresponds to 
$\overline{w}_b$.  
Therefore for any $a \in [1,n]$ we have
\begin{equation}
\begin{split}
\overline{Q}^{(1)}_s &= 
\sum_{k=0}^s \;\;\sum_{1\le b_1 \le \cdots \le b_{s-k} \le a}
\overline{w}_{b_1}\cdots \overline{w}_{b_{s-k}}
\sum_{a< c_1 \le \cdots \le c_k \le n+1}
\overline{w}_{c_1}\cdots \overline{w}_{c_k}\\
&= \sum_{k=0}^s \;\;
s_{(s-k)}(\overline{w}_1,\ldots, \overline{w}_a)
s_{(k)}(\overline{w}_{a+1},\ldots, \overline{w}_n).
\end{split}
\end{equation}
Since the first factor is free from $y_a$ whereas the latter contains it 
as the overall multiplier $y_a^k$,  
the derivative 
$y_a \frac{\partial \overline{Q}^{(1)}_s}{\partial y_a}$ 
coincides with (\ref{k:hs}).
This completes the proof of (\ref{k:mka}) hence that of Proposition \ref{pr:bcd}.
We have finished the proof of Theorem \ref{thk:main}.

\section{Example}\label{sec:ex}

In this section we focus on 
the simplest choice $B^{(r)}_s = B^{(1)}_1$ for the 
set of local states.

\subsection{Explicit formula of the limit shape by Schur functions}

We set $p_a = \pi^{(1)}_1(a)$ for $a \in [1,n+1]$,
where $a$ in the RHS signifies the element of $B^{(1)}_1$ 
corresponding to the semistandard tableau containing $a$ 
in the single box Young diagram.  
So $p_1$ is the density of empty sites and 
$p_a$ with $a \in [2,n+1]$ is the density of balls with color $a$. 
According to (\ref{noc}), one has
$\pi^{(1)}_1(a) = e^{\varpi_1-\alpha_1-\cdots - \alpha_{a-1}}/Q^{(1)}_1$.
Therefore in the regime $\alpha_1, \ldots, \alpha_n >0$ under consideration,
$1>p_1>p_2> \cdots > p_{n+1}>0$ holds.
Of course $p_1+ \cdots + p_{n+1}=1$ should also be satisfied.
For $n=2$ this notation agrees with Example \ref{ex:ca}. 
According to (\ref{noc}) we set 
\begin{align}
p_a = \frac{w_a}{w_1+ \cdots + w_{n+1}}\qquad (1 \le a \le n+1)
\end{align}
in terms of $w_j = z^{-1}_{j-1}z_j$ given in (\ref{k:wde}).
The denominator is $Q^{(1)}_1$ (\ref{qf}).
Thus we find (cf. \cite{KL})
\begin{align}\label{a:zp}
z_a = u^{-\frac{a}{n+1}}p_1p_2\cdots p_{a},
\quad 
u= p_1p_2\cdots p_{n+1}\qquad (0 \le a \le n+1).
\end{align} 
From (\ref{k:ew}), the ball densities $p_1, \ldots, p_{n+1}$ are 
connected to the Young diagram densities 
$\nu_1, \ldots, \nu_n$ as
\begin{align}\label{k:np}
\nu_a = p_{a+1} + p_{a+2} + \cdots + p_{n+1} \qquad (1 \le a \le n).
\end{align}

The equation of state (\ref{est}) reads
\begin{align}\label{k:est3}
z_a \frac{\partial}{\partial z_a}
\log\Bigl(\sum_{j=1}^{n+1}\frac{z_j}{z_{j-1}}\Bigr)
= \delta_{a,1}-\sum_{b=1}^n C_{ab}\nu_b\qquad 
(1 \le a \le n),
\end{align}
where $z_0= z_{n+1}=1$ as in (\ref{k:wde}).
One can easily check that (\ref{k:est3}) is satisfied 
by $z_a$ and $\nu_a$ in (\ref{a:zp}) and (\ref{k:np}) 
provided that $p_1+ \cdots + p_{n+1}=1$ is valid.
This essentially achieves the step (i) in Section \ref{ss:de}.
For the remaining steps (ii) and (iii), 
we have already given the general solution in Theorem \ref{thk:main}.
In the present case the solution $\varepsilon^{(a)}_i  = h^{(a)}_i$ 
can be written down concretely by setting $s=1$ in 
Example \ref{ex:h}:
\begin{align}\label{k:mst}
\varepsilon^{(a)}_i  = \frac{s_{(i^a,1)}(w_1,\ldots, w_{n+1})}
{s_{(i^a)}(w_1,\ldots, w_{n+1})
s_{(1)}(w_1,\ldots, w_{n+1})}\qquad 
(w_a = u^{-\frac{1}{n+1}}p_a).
\end{align}
For simplicity denote the Schur polynomial
$s_\lambda(w_1,\ldots, w_{n+1})$ by $s_\lambda$.
Then the quantity (\ref{k:eta}) is given neatly as 
\begin{align}\label{k:eai}
\eta^{(a)}_i = \varepsilon^{(a)}_i - \varepsilon^{(a)}_{i-1}
= \frac{s_{(i^a,1)}}{s_{(i^a)}s_{(1)}}
-\frac{s_{((i-1)^a,1)}}{s_{((i-1)^a)}s_{(1)}}
= \frac{s_{((i-1)^{a-1})}s_{(i^{a+1})}}
{s_{(i^a)}s_{((i-1)^a)}s_{(1)}},
\end{align}
where we have used a bilinear identity among
the Schur polynomials.

In the simplest case $n=1$,  
the equation of state (\ref{k:est3}) becomes 
$\nu_1= (1+z_1^2)^{-1}$.
From $s_{(i)}(w_1,w_2) = \frac{z_1^{i+1}-z_1^{-i-1}}{z_1-z^{-1}_1}$ and 
$s_{(i,i)}(w_1,w_2)=1$, 
the result (\ref{k:eai}) reduces to 
\begin{align}\label{k:sss}
\eta^{(1)}_i = \frac{1}{s_{(i)}s_{(i-1)}s_{(1)}}
= \frac{\zeta^i(\zeta-1)^2}
{(\zeta+1)(\zeta^{i}-1)(\zeta^{i+1}-1)},\qquad
\zeta = z_1^{-2}.
\end{align}
This agrees with a corresponding result in \cite{KL}.

\subsection{Scaling behavior of the width of the Young diagrams}

Note that $\eta^{(a)}_i$ (\ref{k:eta}) 
is the $1/L$ scaled length of the $i$-th column of the Young diagram $\mu_a$.
Moreover we have used the scaling behavior (\ref{mpel}).
Thus the above result (\ref{k:eai}) should be understood to be effective 
in the range $1 \le i \le I_a$ where
\begin{align}\label{k:Ja}
L \eta^{(a)}_{I_a} \simeq 1\qquad (L \gg 1).
\end{align}
This yields a crude estimate of the scaling behavior of the 
{\em width} $I_a$ of the Young diagram $\mu_a$ 
as $L$ grows large.

Let us investigate the consequence of (\ref{k:Ja}) closely for the 
regime $1> p_1 > p_2 > \cdots > p_{n+1}>0$.
From (\ref{k:sc}) and $w_1 > w_2 \cdots > w_{n+1}>0$, we see that  
$s_{(i^a)}= s_{(i^a)}(w_1,\ldots, w_{n+1})$ behaves as  
\begin{align*}
s_{(i^a)}=
\frac{\Delta(w_1,\ldots, w_a)\Delta(w_{a+1},\ldots, w_{n+1})}
{\Delta(w_1,\ldots, w_{n+1})}
(w_1\cdots w_a)^{i+n-a+1}\bigl(1+ O(e^{-c i})\bigr) \qquad
(i \rightarrow \infty)
\end{align*}
for some constant $c>0$, where
$\Delta(x_1,\ldots, x_m) = \prod_{1 \le j < k \le m}(x_j-x_k)$
is the Vandermonde determinant.
Applying this to (\ref{k:eai}) and using (\ref{a:zp}) we find 
that $\eta^{(a)}_i$ tends to 0 as $i \rightarrow \infty$ as 
\begin{align}\label{k:ei}
\eta^{(a)}_i \simeq 
\frac{\prod_{j=1}^a(p_j-p_{a+1})
\prod_{j=a+1}^{n+1}(p_a-p_j)}
{\prod_{j=1}^{a-1}(p_j-p_a)\prod_{j=a+2}^{n+1}(p_{a+1}-p_j)}
\frac{1}{p_a}\Bigl(\frac{p_{a+1}}{p_a}\Bigr)^{i+n-a}
\quad (i \rightarrow \infty)
\end{align}
up to exponentially small corrections.
Thus the estimate (\ref{k:Ja}) implies the logarithmic scaling 
\begin{align}\label{k:ia}
I_a \simeq \frac{\log L}{\log \frac{p_a}{p_{a+1}}}\qquad 
(L \rightarrow \infty)
\end{align}
in the leading order.
For $n=a=1$ and $1-p_1=p_2=p$, (\ref{k:ei}) and (\ref{k:Ja}) 
lead to
$I_1 \simeq \log\Bigl(\frac{(1-2p)^2L}{1-p}\Bigr)/\log\frac{1-p}{p}$.
Incidentally this reproduces $\mu_L$ in \cite[Th.2(i)]{LLP}
including the coefficient.

The result (\ref{k:ia}) indicates yet another scaling behavior at 
$p_a = p_{a+1}$.
For simplicity let us consider the most degenerate case of such situations 
$p_a = \frac{1}{n+1}$ for all $a \in [1,n+1]$. 
It corresponds to the completely 
random distribution of the balls and empty sites.
Then we have $w_a=1$ and (\ref{k:eai}) simplifies to 
\begin{align}
\eta^{(a)}_i = \frac{a(n+1-a)}{(n+1)(i+a-1)(i+a)}.
\end{align}
Therefore the estimate (\ref{k:Ja}) gives 
\begin{align}
I_a \simeq \sqrt{\frac{a(n+1-a)L}{n+1}}.
\end{align}
This square root scaling behavior is a signal of criticality
as observed in \cite{LLP}.

\subsection{Numerical check}

Here we deal with the $n=2$ case, i.e., 2-color BBS.
The relevant KR crystals are $B^{(1)}_l$ and $B^{(2)}_l$.
We parametrize the both by the set 
$\{x=(x_1,x_2,x_3)\in (\Z_{\ge 0})^3 \mid x_1+x_2+x_3=l\}$.
For $B^{(1)}_l$,  $x_i$ is the number of letter $i$ contained 
in the semistandard tableau on the Young diagram of single row shape 
with length $l$.
For $B^{(2)}_l$,
$x_1, x_2, x_3$ are the number of the columns  
 ${2 \atop 3},{1 \atop 3},{1 \atop 2}$ 
in the semistandard tableau on the Young diagram of double row shape 
with length $l$.
The combinatorial 
$R: x\otimes y \mapsto \tilde{y} \otimes \tilde{x}$ 
and the local energy $H(x\otimes y)$ necessary 
to compute $E^{(1)}_l$ and $E^{(2)}_l$ are those acting on 
$B^{(1)}_l \otimes B^{(1)}_1$ and 
$B^{(2)}_l \otimes B^{(1)}_1$, respectively.
They are given explicitly by piecewise linear formulas in 
\cite[sec.2.2]{KOY}$|_{n=3}$ with a slight conventional adjustment.
We summarize them in  Table 1.

\begin{table}[H]\label{tab:1}
\begin{tabular}{c|c|c}

& Combinatorial $R$ & local energy $H(x\otimes y)$\\
\hline
$E^{(1)}_l$ & $R|_{m=1}$ (2.1)  & $Q_0(x,y)$ \\
$E^{(2)}_l$ & ${}^\vee\!R|_{m=1}$ (2.3)  & $P_{-1}(x,y)$ 

\end{tabular}
\caption{Notations and equation numbers except in the first column 
are those in \cite{KOY}. }
\end{table}

We have generated a BBS state in $\{1,2,3\}^L$
with a prescribed ball densities $1>p_1>p_2>p_3>0$ 
and length $L=1000$ by computer.
Calculating the energy $E^{(1)}_l, E^{(2)}_l$ by (\ref{ok:Eai})
we extract the Young diagrams $\mu_1$ and $\mu_2$.
After scaling by $1/L$ vertically we plot them (called ``BBS") 
together with the prediction (\ref{k:eai}) (called ``TBA") 
in Figure 1 and 2 below.

\begin{figure}[H]\label{fig1}
\begin{tabular}{c}
\begin{minipage}[t]{0.45\hsize}
\begin{center}
\includegraphics[scale=1]{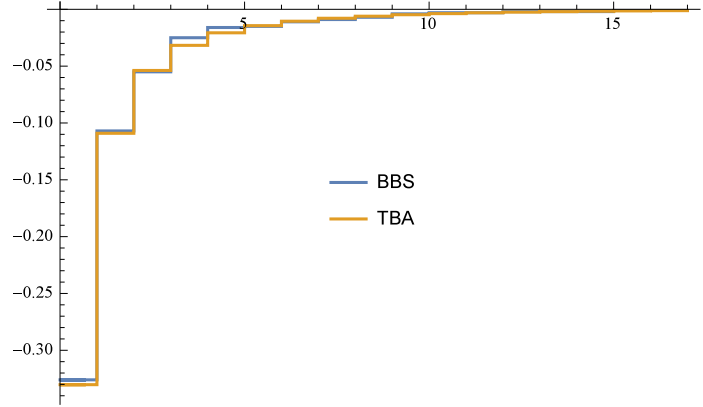}
\end{center}
\end{minipage}
\end{tabular}
\caption{Vertically $1/L$ scaled Young diagram 
$\mu_1$. $L=1000, (p_1, p_2, p_3)=(\frac{7}{18}, \frac{6}{18}, \frac{5}{18})$.} 
\end{figure}
\begin{figure}[H]\label{fig2}
\begin{tabular}{c}
\begin{minipage}[t]{0.45\hsize}
\begin{center}
\includegraphics[scale=1]{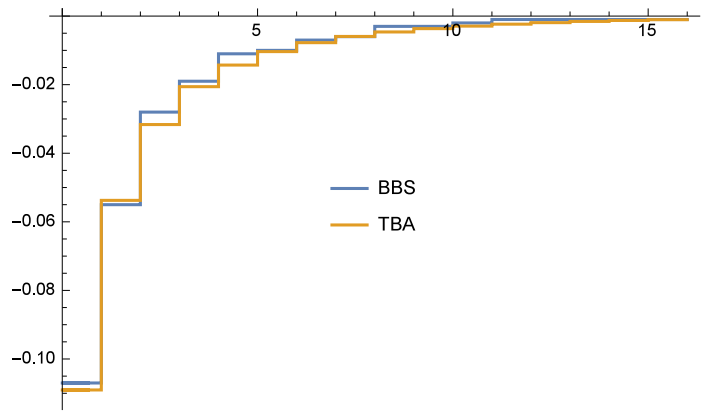}
\end{center}
\end{minipage}
\end{tabular}
\caption{Vertically $1/L$ scaled Young diagram $\mu_2$. $L=1000, 
(p_1, p_2, p_3)=(\frac{7}{18}, \frac{6}{18}, \frac{5}{18})$.} 
\end{figure}

According to (\ref{k:Ja}) we have truncated the scaled 
$\mu_1$ and $\mu_2$ at the width $I_1=17$ and $I_2=16$.
The agreement of the numerical data from BBS 
and the TBA prediction is more or less satisfactory.

\section{Discussion}\label{sec:dis}

\subsection{Summary}
We have elucidated a new interplay among 
the randomized BBS, Markov processes of carriers, 
KR modules/crystals, combinatorial $R$, local energy, 
deformed characters, Fermionic formulas,
rigged configurations, Q and Y-systems, 
TBA equations and so forth.
Our main result is Theorem \ref{thk:main} which identifies  
the stationary local energy of the KR crystal (\ref{k:izm}) as 
the explicit solution to the difference equation (\ref{con}) originating from TBA.
It determines the equilibrium shape of the Young diagrams 
$\mu_1,\ldots, \mu_n$ in the scaling limit as in 
(\ref{k:es5}), (\ref{k:es}) and (\ref{k:eta}).
These random Young diagrams arise as the conserved quantities
(generalized soliton contents) of the randomized BBS and obey
the probability distribution given by the Fermionic form (\ref{prom}).

\subsection{Generalization to simply-laced case}
 
Although the above results 
are concerned with the quantum affine algebra $U_q(\hat{\mathfrak{g}})$
with $\hat{\mathfrak{g}}=A^{(1)}_n$,
all the essential ingredients are known or at least 
conjecturally/conceptually ready for general quantum affine algebras. 
In particular formulas for the simply-laced cases
$\hat{\mathfrak{g}} = A^{(1)}_n, D^{(1)}_n, E^{(1)}_{6,7,8}$
possess a quite similar and simple structure.
Let us explain them briefly and 
conjecturally describe the parallel results on 
the randomized BBS of type $ADE$ uniformly. 
We set $\mathfrak{g} = A_n, D_n, E_{6,7,8}$ 
according to $\hat{\mathfrak{g}} = A^{(1)}_n, D^{(1)}_n, E^{(1)}_{6,7,8}$
and $n=6,7,8$ for $\hat{\mathfrak{g}} = E^{(1)}_{6,7,8}$.
The matrix $(C_{ab})_{1 \le a,b \le n}$ is to be understood as the 
Cartan matrix of $\mathfrak{g}$ and 
the relation $a\sim b$ is defined by $C_{ab}=-1$.

The KR modules $\{W^{(r)}_s \mid (r,s) \in [1,n] \times \Z_{\ge 0} \}$ 
over $U_q(\hat{\mathfrak{g}})$
are specified in terms of the Drinfeld polynomials.
See for example \cite[sec.4.2]{KNS}.
The corresponding KR crystals 
$\{B^{(r)}_s \mid (r,s) \in [1,n]\times \Z_{\ge 0} \}$ 
have been constructed for 
$\hat{\mathfrak{g}}=D^{(1)}_n$ in \cite{OS}, whereas 
their existence is yet conjectural in general for 
$\hat{\mathfrak{g}} = E^{(1)}_{6,7,8}$.
Here we assume them and denote by 
$u^{(r)}_s \in B^{(r)}_s$ the unique element of weight $s\varpi_r$.
The local energy $H$ should be so taken as $H \in \Z_{\ge 0}$ and  
$H(u^{(a)}_i \otimes u^{(r)}_s) =0$ on any $B^{(a)}_i \otimes B^{(r)}_s$.
Set $Q^{(a)}_i = \sum_{b \in B^{(a)}_i} e^{\mathrm{wt}(b)}$.
Although it is no longer a character of an  irreducible
$\mathfrak{g}$ module in general, it satisfies the Q-system (\ref{qsys}).
See for example \cite{HKOTY} and \cite[sec.13]{KNS} and references therein.

The $U_q(\hat{\mathfrak{g}})$ BBS is formulated 
in the same manner as Section \ref{sec:bbs}.
Take the set of local states to be $B^{(r)}_s$.
We consider the randomized $U_q(\hat{\mathfrak{g}})$ BBS 
where the local states 
and the stationary measure of the carrier for the time evolution 
$T^{(a)}_i$ are given as Proposition \ref{pr:st}.
Then the stationary local energy $h^{(a)}_i$ 
takes the same form as (\ref{k:ep}) 
or equivalently (\ref{k:izm}).

Concerning the deformed character (\ref{k:xw}),
the corner transfer matrix energy $D$ in (\ref{k:d}) needs to be replaced by 
\begin{align}\label{k:db}
D(b_1 \otimes \cdots \otimes b_L)  = 
\sum_{j=1}^LH(b^\natural_j\otimes b^{(1)}_j) +
\sum_{1 \le i < j \le L}H(b_i \otimes b^{(i+1)}_j),
\end{align}
where $b^\natural_j \in B^{(k_j)}_{l_j}$ is the unique element such that 
$\varphi(b^\natural_j)= l_j \Lambda_0$.
See \cite[sec.5.1]{O} for a detailed account of this.
The first sum on the RHS is referred to as the boundary energy. 
It is $0$ for $A^{(1)}_n$ but is nontrivial for the other types.

The Fermionic form $M(B,\lambda,w)$ is defined 
by the same formulas as (\ref{mdef})--(\ref{wc}).
Theorem \ref{th:kss}\footnote{For general affine Lie algebra it is often called 
$X=M$ conjecture \cite{HKOTY,HKOTT} .} 
is valid for $D^{(1)}_n$ \cite{OSSS, N} 
and conjecturally valid for $E^{(1)}_{6,7,8}$.
Proposition \ref{pr:mq} has been shown in \cite{HKOTY}.
Lemma \ref{le:wick} is influenced by the boundary energy and replaced by 
\begin{equation}\label{wick2}
\begin{split}
&\frac{\partial}{\partial w} \log \chi_w(B_1\otimes \cdots \otimes B_L)|_{w=1}
\\
&\;\; = \sum_{1 \le i<j \le L}\frac{1}{Q_i Q_j}
\sum_{b \otimes c\in B_i \otimes B_j}
H(b \otimes c) e^{\mathrm{wt}(b \otimes c)}
+ \sum_{1\le i \le L}\frac{1}{Q_i}
\sum_{c \in B_i}H(b^\natural_i \otimes c)  e^{\mathrm{wt}(c)}.
\end{split}
\end{equation}

As for the TBA analysis, all the relations 
from (\ref{prom}) until (\ref{k:sn}) remain unchanged\footnote{The unique  
exception is the last expression in (\ref{k:hro}) which is specific to type $A_n$.}.
In particular, the property 
$\lim_{l \rightarrow \infty}
(Q^{(a)}_{l+1}/Q^{(a)}_l)= e^{\varpi_a}$ used to 
simplify (\ref{qoq}) is valid not only for 
non-exceptional cases \cite[Th. 7.1 (C)]{HKOTY} 
but for all types \cite[Prop.5.9]{He}.

We conjecture that Theorem \ref{thk:main} is also valid for 
type $D_n$ and $E_{6,7,8}$.
In fact admitting Theorem \ref{th:kss}, 
it can be shown that $\varepsilon^{(a)}_i = h^{(a)}_i$ 
provides a solution 
to the difference equation (\ref{con})\footnote{This assertion is the analogue 
of Proposition \ref{pr:de}, which was the ``first half" of Theorem \ref{thk:main}.}.
The proof uses (\ref{wick2}).
In particular with the notation (\ref{k:sn}),  its $L=2$ case 
captures the stationary local energy $h^{(a)}_i$ as
\begin{align}
\frac{\partial \log(Q^{(a)}_i \ast Q^{(r)}_s)}{\partial w} 
\Big|_{w=1} = 
h^{(a)}_i + \frac{1}{Q^{(a)}_i}
\sum_{c \in B^{(a)}_i}H(b^\natural_1 \otimes c)  e^{\mathrm{wt}(c)}
+ \frac{1}{Q^{(r)}_s}
\sum_{c \in B^{(r)}_s}H(b^\natural_2 \otimes c)  e^{\mathrm{wt}(c)}.
\end{align}
The point is that the effect of extra ``boundary terms" 
containing $b^\natural_1 \in B^{(a)}_i,
b^\natural_2 \in B^{(r)}_s$ is canceled by those in (\ref{wick2}), leaving  
the difference equation unchanged from (\ref{con}).
As for the boundary condition for the difference equation,
we conjecture that (\ref{k:otm}) or equivalently (\ref{k:mka}) 
holds universally for type $D_n$ and $E_{6,7,8}$.
It is an intriguing relation involving the local energy whose 
proof will shed new light into the KR crystals and the Q-system.

\subsection{Further outlook}
We expect the generalization to the non simply-laced cases
and twisted affine Lie algebras is 
also feasible albeit with a slight technical complexity.
Another obvious direction of a future research is periodic systems.
The generalized BBS for $A^{(1)}_n$ with the quantum space 
$(B^{(1)}_1)^{\otimes L}$ has been studied 
under the periodic boundary condition \cite{KT}.
It also has an $n$-tuple of Young diagrams 
as a label of iso-level sets for which 
a Fermionic formula \cite[eq.(57), Th.3]{KT} for the multiplicity 
has been obtained under a technical assumption.
It will be interesting to analyze it by TBA similarly to this paper.
In the simplest case $n=1$, the Fermionic formula has been fully  
justified and reduces to 
\begin{align}\label{k:pf}
\frac{L}{L-2M}\prod_{i\ge 1}\binom{p^{(1)}_i+m^{(1)}_i-1}{m^{(1)}_i}
\end{align}
for system size $L$ and $M$-ball sector with $M<\frac{L}{2}$
in the same notation as (\ref{pai})\footnote{$p^{(1)}_i = 
L - 2\sum_{j\ge 1}\min(i,j)m^{(1)}_j,\,M= \sum_{j\ge 1} j m^{(1)}_j$.
One can show $(\ref{k:pf}) \in \Z_{\ge 0}$.}.
So at least in this simplest situation, 
the scaled limit shape of the Young diagram remains the same as (\ref{k:sss}).

There are a number of further challenging problems
to be investigated.
We list a few of them as closing remarks.

\begin{enumerate}

\item Study the limit shape problem
when the BBS states are inhomogeneous as   
$B^{(r_1)}_{s_1} \otimes \cdots \otimes B^{(r_L)}_{s_L}$
with a given statistical distribution of $(r_i, s_i)$.

\item Can one architect a BBS like dynamical system 
whose Markov process of carriers has the stationary measure
described by $q$-characters \cite{FR}?

\item Can one extend the TBA analysis so as to include 
$w$-binomials in (\ref{prom}) with $w \neq 1$? 
What is the counterpart of the BBS corresponding to such a generalization?
 
\item Our TBA analysis in this paper was spectral parameter free.
See the remark after (\ref{ser}).
Is there any Yang-Baxterization of Theorem \ref{thk:main}?

\end{enumerate}

\subsection*{Acknowledgements}
The authors thank Rei Inoue, Makiko Sasada and Satoshi Tsujimoto for kind interest.
This work is supported by 
Grants-in-Aid for Scientific Research No.~18H01141 from JSPS.

\end{document}